\documentclass{article}
\usepackage{fullpage}
\usepackage{color}

\usepackage{graphicx}
\usepackage{amsmath,amssymb,amsfonts,amsthm}
\usepackage{mathalfa}
\usepackage{tikz}
\usetikzlibrary{calc}
\usepackage{color}

\newenvironment{myproof}[1]{\smallskip\noindent{\sc Proof #1.}}%
        {\hspace*{\fill}$\Box$\par}

\usepackage{enumitem}
\newlist{ass}{enumerate}{1}
\setlist[ass]{label = {\bf (A\arabic*)}, resume}
\newlist{assFP}{enumerate}{1}
\setlist[assFP]{label = {\bf (A\arabic*-P)}, resume}
\newlist{assBV}{enumerate}{1}
\setlist[assBV]{label = {\bf (A\arabic*-BV)}, resume}
\newlist{hyp}{enumerate}{1}
\setlist[hyp]{label = {\bf (H\arabic*)}, resume}
\newlist{obj}{enumerate}{1}
\setlist[obj]{label = {\bf (O\arabic*)}, resume}
\newlist{lyap}{enumerate}{1}
\setlist[lyap]{label = {\bf (L\arabic*)}, resume}

\theoremstyle{plain}
\newtheorem{lem}{Lemma}
\newtheorem{prop}{Proposition}
\newtheorem{thm}{Theorem}

\theoremstyle{definition}
\newtheorem{defn}{Definition}

\theoremstyle{remark}
\newtheorem{rem}{Remark}

\usepackage{graphicx}
\usepackage{amsmath,amssymb,amsfonts,yhmath}
\usepackage{caption}
\usepackage{subcaption}

\usepackage{color}
\usepackage{tikz}
\usetikzlibrary{circuits.ee.IEC}
\usetikzlibrary{arrows}

\newcommand{\me}{\mathtt{e}}

\newcommand{\R}{\mathbb{R}}

\newcommand{\cA}{\mathcal{A}}
\newcommand{\cB}{\mathcal{B}}
\newcommand{\cC}{\mathcal{C}}
\newcommand{\cD}{\mathcal{D}}

\newcommand{\cF}{\mathcal{F}}

\newcommand{\cH}{\mathcal{H}}

\newcommand{\cJ}{\mathcal{J}}
\newcommand{\cK}{\mathcal{K}}
\newcommand{\cL}{\mathcal{L}}

\newcommand{\cQ}{\mathcal{Q}}

\newcommand{\cS}{\mathcal{S}}
\newcommand{\cT}{\mathcal{T}}

\newcommand{\cW}{\mathcal{W}}

\DeclareMathOperator{\dom}{dom}

\DeclareMathOperator{\col}{col}

\newcommand{\eps}{\varepsilon}

\newcommand{\diag}{\text{\normalfont diag}}

\newcommand{\innProd}[2]{\left\langle #1, #2 \right\rangle}

\usepackage{mathtools} 
\DeclarePairedDelimiter\floor{\lfloor}{\rfloor}

\newcommand{\footremember}[2]{%
   \footnote{#2}
    \newcounter{#1}
    \setcounter{#1}{\value{footnote}}%
}
\newcommand{\footrecall}[1]{%
    \footnotemark[\value{#1}]%
} 

\title
{Disturbance-to-State Stabilization and Quantized Control for Linear Hyperbolic Systems}
\author{%
    Aneel Tanwani\footremember{corr}{Corresponding author's email: {\tt aneel.tanwani@laas.fr}}\footremember{laas}{LAAS--CNRS, Universit\'e de Toulouse, CNRS, 7 Avenue du Colonel Roche, 31400 Toulouse, France.    
}%
    \and Christophe Prieur\footremember{gipsa}{Gipsa Lab--CNRS, University of Grenoble Alpes, CNRS, 11 Rue des Math\'ematiques, BP 46, 38402 Saint Martin d'H\`eres, France.
}%
  \and Sophie Tarbouriech\footrecall{laas}
}
\date{}

\begin{document}
\maketitle

\begin{abstract}
We consider a system of linear hyperbolic PDEs where the state at one of the boundary points is controlled using the measurements of another boundary point.
Because of the disturbances in the measurement, the problem of designing dynamic controllers is considered so that the closed-loop system is robust with respect to measurement errors. 
Assuming that the disturbance is a locally essentially bounded measurable function of time, we derive a disturbance-to-state estimate which provides an upper bound on the maximum norm of the state (with respect to the spatial variable) at each time in terms of $\cL^\infty$-norm of the disturbance up to that time.
The analysis is based on constructing a Lyapunov function for the closed-loop system, which leads to controller synthesis and the conditions on system dynamics required for stability.
As an application of this stability notion, the problem of quantized control for hyperbolic PDEs is considered where the measurements sent to the controller are communicated using a quantizer of finite length.
The presence of quantizer yields practical stability only, and the ultimate bounds on the norm of the state trajectory are also derived.
\end{abstract}

\section{Introduction}
Partial differential equations (PDEs), or distributed parameter systems, have appeared as a tool for modeling several complex physical phenomena, and there is now a considerable literature on analysis and simulation of such systems.
More recently, over the past decade, there has been a surge in control community for designing control algorithms for PDEs so that their behavior can be steered towards some desired performance level.
This has led the researchers to generalize several control-theoretic questions from the finite-dimensional systems in the context of infinite-dimensional systems.
In that spirit, this article formulates a robust stability notion when the measurements used for feedback control in hyperbolic PDEs are subjected to unknown disturbances.
In the literature on ordinary differential equations (ODEs), the property of {\em input-to-state stability} (ISS) induces this desired robust behavior while regarding the disturbances as exogenous inputs in the closed-loop system.
The Lyapunov function based techniques available for verifying ISS are thus generalized in the context of hyperbolic PDEs in this article.

Hyperbolic PDEs represent a class of such infinite dimensional systems, which have been used in modeling physical system such as shallow water equations, and also to model time-delays in engineering systems. Several results on analysis, Lyapunov stability, and feedback control design of hyperbolic systems have been published in the recent past, see the book \cite{BastCoro16} for an overview of results.

\subsection{System Class}
We consider the feedback control for the class of linear hyperbolic PDEs described by the equation
\begin{subequations}\label{eq:sysHyp}
\begin{equation}\label{eq:sysHyp:dyna}
\frac{\partial X}{\partial t} (z,t) +\Lambda \frac{\partial X}{\partial z} (z,t) = 0
\end{equation}
where $z \in [0,1]$ is the spatial variable, and $t \in \R_+:= [0,\infty)$ is the time variable. The matrix $\Lambda$ is assumed to be diagonal and positive definite. We call $X:[0,1] \times \R_+ \rightarrow \R^n$ the state trajectory.
The initial condition is defined as
\begin{equation}
X(z,0) = X^0(z), \quad z \in (0,1)
\end{equation}
\end{subequations}
for some function $X^0 : (0,1) \rightarrow \R^n$.
The value of the state $X$ is controlled at the boundary $z=0$ through some input $u: \R_+ \to \R^{m}$ so that
\begin{equation}\label{eq:initCond}
X (0,t) = H X(1,t) + B u(t)
\end{equation}
where $H \in \R^{n \times n}$ and $B \in \R^{n \times m}$ are constant matrices.
The system~\eqref{eq:sysHyp}-\eqref{eq:initCond} forms a class of 1-D boundary controlled hyperbolic PDEs, for which several fundamental results can be found in \cite{BastCoro16}.

We consider the case when only the measurement of the state $X$ at the boundary point $z=1$ is available for each $t \ge 0$, and this measurement is subjected to some bounded disturbance.
We thus denote the output of the system by
\begin{equation}\label{eq:defOut}
y(t) = X(1,t) + d(t)
\end{equation}
where the disturbance $d \in \cL^\infty([0,\infty),\R^n)$ may arise due to low resolution of the sensors, uncertain environmental factors, or errors in communication.

We are interested in designing a feedback control law $u$ as a function of the output measurement $y$, that is $u = \cF(y)$ for some operator $\cF$, which stabilizes the system in some appropriate sense, and the behavior of the closed-loop system is robust with respect to the measurement disturbances.
Here, we allow the possibility that $u$ may be obtained via a dynamic compensator so that $\cF$ is an operator with memory.
In particular, it is desired that the closed-loop trajectories satisfy the following disturbance-to-state stability (DSS) estimate:
\begin{equation}\label{eq:maxEstGen}
\max_{z\in [0,1]} \vert X(z,t) \vert \le c\, \me^{-at} M_{X^0} + \gamma \left( \|d_{[0,t]}\|_\infty \right)
\end{equation}
for some constants $a, c, M_{X^0} > 0$, and $\gamma$ a class $\cK_\infty$ function. Here, $\|d_{[0,t]}\|_\infty$ denotes the essential supremum of $\vert d(s) \vert$ for $s$ contained in $[0,t]$, and for given $z$ and $t$, $\vert X(z,t) \vert$ denotes the usual Euclidean norm of $X(z,t) \in \R^n$. The constant $M_{X^0}$ is such that it depends on some norm associated with the function $X^0$ and possibly the initial state chosen for the dynamic compensator $u$. The DSS property ensures that in the absence of disturbance, that is $d \equiv 0$, the maximum norm of $X$ (with respect to spatial variable) decreases exponentially in time with a uniform decay rate. In the presence of nonzero disturbances, that is $d \not\equiv 0$, the maximum value of $X$ over $[0,1]$, at each time $t \ge 0$, is bounded by the maximum norm of the disturbance over the interval $[0,t]$ and an exponentially decaying term due to the initial condition of the system. Due to the semigroup property, the conditions we impose on the system to obtain estimate \eqref{eq:maxEstGen}, also ensure that if $d(t) \to 0$, then $\max_{z\in [0,1]} \vert X(z,t) \vert$ also converges to zero with time, see Remark~\ref{rem:sgProp} and Section~\ref{sec:vanish}.

It turns out that the function (of initial state) $M_{X^0}$ that we compute to establish \eqref{eq:maxEstGen} is such that, even if $X^0 \equiv 0$, $M_{X^0}$ is not necessarily equal to zero. However, using \eqref{eq:maxEstGen}, we can obtain an alternate estimate of the form
\begin{equation}\label{eq:genIsps}
\max_{z\in [0,1]} \vert X(z,t) \vert \le c \, \me^{-at} \! \! \max_{z\in [0,1]} \vert X^0(z,t) \vert + \gamma  (\|d_{[0,t]}\|_\infty) + C \me^{-at}
\end{equation}
for some $C>0$. This estimate guarantees attractivity of the origin $X = 0$ in $\cC^0([0,1],\R^n)$, and only {\em practical} stability.

Drawing comparisons from the literature on stability of finite-dimensional systems, it is observed that the estimates of the form \eqref{eq:maxEstGen} and \eqref{eq:genIsps} have been studied under the notion of input-to-state stability (ISS), pioneered in \cite{Sontag89}, and more generally input-to-state practical stability (ISpS) \cite{JianTeel94}, respectively.
One of the most fundamental results in the ODEs literature, which makes the ISS property extremely useful for design problems, is that the ISS estimates can be equivalently characterized in terms of Lyapunov dissipation inequalities.
In our approach, we also propose a controller design which allows us to construct a Lyapunov function for the closed-loop such that the corresponding dissipation inequality is of the same form as in the finite-dimensional case. This proves to be sufficient for arriving at the estimate \eqref{eq:maxEstGen}.



\subsection{Motivation}\label{sec:mot}
The motivation for studying the DSS property comes from the application in quantized control.
When the measurement $X(1,t)$ can not be passed precisely to the controller, but has to be encoded using finitely many symbols, one can see $d$ in \eqref{eq:defOut} as the error between the actual value and the quantized value of the signal $X(1,t)$. The quantizers are typically designed to operate over a compact set in the output space.
Within this operating region, the quantization error remains constant and hence one expects the state trajectory to converge to a ball around the origin parameterized by the size of quantization error. Hence, to obtain this practical stability, the controller must ensure that the state trajectory remains within the compact set for which the quantizer is designed.

To implement this methodology in the context of PDEs under consideration, the problem is to find a controller which ensures the DSS estimate \eqref{eq:maxEstGen} holds and that the output $X(1,\cdot)$ remains within the range of the quantizer. The DSS estimate also ensures practical stability in this setup since the $X(z,\cdot)$ eventually converges to a ball around the origin whose radius is parameterized  by the sensitivity of the quantizer.

\subsection{Literature Overview}
In case there are no perturbations, that is, $d \equiv 0$, one typically chooses $u(t) = K y(t)$ such that the closed-loop boundary condition
\begin{equation}\label{eq:initCondLoop}
X(0,t) = (H+BK) X(1,t)
\end{equation}
satisfies a certain dissipative condition.
This control law yields asymptotic stability of the system with respect to $\cH^2$-norm \cite{CoronBastin08}, or $\cC^1$-norm \cite{CoronBastin15}, depending on the dissipativity criterion imposed on $H+BK$.
In the presence of perturbations $d \not \equiv 0$, one has to modify the stability criteria as the asymptotic stability of the origin can no longer be established.

One finds the Lyapunov stability criteria with $\cL^2$-norm and dissipative boundary conditions in \cite{BastinCoronAndreaIFAC08}. Lyapunov stability in $\cH^2$-norm for nonlinear systems is treated in \cite{CoronBastin08}. Thus, the construction of Lyapunov functions in $\cH^2$-norm for the hyperbolic PDEs with static control laws can be found in the literature.
Because our controller adds dynamics to the closed-loop, the basic idea behind the construction of Lyapunov function for the closed-loop system is to use the ISS property of the hyperbolic PDE and the controller dynamics.

In the literature, one finds various instances where the ISS related tools are used for stability analysis of interconnected systems.
In the paper \cite{ito2012capability}, an integral ISS Lyapunov function is computed for a networks described by a finite-dimensional nonlinear function.
Small gain theorem is crucial when interconnecting ISS systems as exploited in \cite{geiselhart2012numerical,dashkovskiy2007iss}.

For infinite dimensional systems, the problem of ISS has attracted attention recently but most of the existing works treat the problem with respect to uncertainties in the dynamics. See, for example \cite{mironchenko2015note}, where a class of linear and bilinear systems is studied.
See also \cite{dashkovskiy2013input} where a linearization principle is applied for a class of infinite-dimensional systems in a Banach space.
When focusing on parabolic partial differential equations, some works to compute ISS Lyapunov functions have also appeared, such as \cite{mironchenko2015construction,mazenc2011strict}.
For time-varying hyperbolic PDEs, construction of ISS Lyapunov functions has also been addressed in \cite{PrieurMazenc:hyper:11}.

For hyperbolic systems, when seeking robust stabilization with measurement errors, one could see that the results in \cite{EspiGira16} provide robust stability of $X(\cdot,t)$ in $\cL^2((0,1);\R^n)$ space by using static controllers and piecewise continuous solutions. However, the DSS estimate \eqref{eq:maxEstGen} requires stability in $\cC^0([0,1];\R^n)$ space equipped with maximum norm.

\subsection{Contribution}
For PDEs in general, the results on stability with respect to measurement errors have not yet appeared in the literature; The only exception being the recent work reported in \cite{KaraKrst16, KaraKrst17} which derives ISS bounds for 1-D parabolic systems in the presence of boundary disturbances but without the use of Lyapunov-based techniques. Such questions have remained unaddressed for hyperbolic PDEs, which is the topic of this paper.
Furthermore, the paper also includes a design element in the sense that the controller that achieves the DSS property is also being synthesized.
On the other hand, the problem of quantized control has mostly been studied in finite-dimensional systems so far \cite{Libe03Aut, NairFagn07, PrieTanw17, TanwPrie16b}, and this paper extends this problem setting to the case of PDEs. While there exist some works on quantized control of finite-dimensional systems in the presence of delays \cite{Libe06}, the model of hyperbolic PDEs treated in this paper is much more general and as such no direct comparison can be drawn between the earlier approaches and the techniques developed in this paper.

To achieve the aforementioned objectives, we propose to use a dynamic controller instead of a static one, as proposed in the conference version of this article~\cite{TanwPrie16}.
The reason for emphasizing the use of dynamic controllers is that we are looking for a way to bound $|X(z,t)|$, for each $z \in [0,1]$, which in our knowledge is only possible if a bound on the $\cH^1$-norm of $X(\cdot,t)$ is obtained, see Section~\ref{sec:prelim} for an explanation.
Existence of solutions $X$ in the space $\cH^1((0,1);\R^n)$ requires us to use inputs which are at least absolutely continuous.
If we allow perturbations $d$ to be discontinuous, static controllers would not yield smooth enough solutions.
The dynamic controller is therefore added to smoothen the discontinuity effect of the perturbations.

The addition of dynamic controller introduces a coupling of ODEs and PDEs in the closed loop which makes the analysis more challenging. Results on well-posedness of such coupled systems are proposed resulting in certain regularity of the closed-loop solutions, which is important to obtain appropriate estimates.
We use Lyapunov function based analysis to synthesize the controller and guarantee DSS with respect to the perturbation $d$.
The results are then used to study the application of quantized control:
We establish practical stability of the system, and derive ultimate bounds on the state trajectory in terms of the quantization error.
Compared to the conference article~\cite{TanwPrie16}, we provide rigorous mathematical proofs of the main results. The stability notions treated in the paper are more general, and several discussions related to connections with other stability notions are also included. Moreover, this article rigorously establishes the existence of solution for the closed-loop system in Theorem~\ref{thm:sol}, which was not addressed in \cite{TanwPrie16}.


\section{Refined Problem Formulation}\label{sec:prelim}

In this section, we recall some preliminaries associated with the solution space adopted for hyperbolic PDEs in our framework. Connections between the DSS notion and the norm associated with the solution space of the PDE are made explicit. Finally, the idea of dynamic controller is proposed to guarantee solutions with appropriate regularity.

\subsection{Preliminaries}
For a function $X:(0,1) \rightarrow \R^n$, we denote its gradient by $\partial X$ or $X_z$, and for $X:(0,1)\times \R_+ \to \R^n$, we denote the gradient with respect to first argument by $\partial_z X$, or $X_z$, and the gradient with respect to second argument by $\partial_t X$, or $X_t$ with the obvious interpretation that $z$ and $t$ denote the spatial and time variable, respectively. The space $\cW^{k,p}((0,1);\R^n)$ comprises functions for which the $k$-th derivative, denoted $\partial^k X$, exists and $\partial^k X \in \cL^p((0,1);\R^n)$. We use the shorthand $\cH^1$ for the space $\cW^{1,2}$. The space $\cH^1$ is naturally equipped with the $\cH^1$-norm defined as:
\[
\| X \|_{\cH^1((0,1);\R^n)} := (\| X \|_{\cL^2((0,1);\R^n)}^2 + \| \partial X \|_{\cL^2((0,1);\R^n)}^2)^{1/2}.
\]
In literature on stability analysis of hyperbolic PDEs, we find several notions of stability depending on the norm with which the solution space is equipped.
If we choose to control the $\cL^2$-norm of the state trajectory only, the problem is that it doesn't yield any bounds on $\max_{z\in[0,1]} \vert X(z,t) \vert$, for a given $t\ge 0$.

\subsection{Obtaining DSS using $\cH^1$-norm}
The motivation for introducing the $\cH^1((0,1);\R^n)$ solution space can be seen in the following proposition:
\begin{prop}\label{prop:maxH1}
Given any function $X:[0,1] \to \R^n$ such that $X \in \cC^0([0,1];\R^n) \cap \cH^{1}((0,1);\R^n)$. It holds that, for every $z \in [0,1]$,
\begin{equation}\label{eq:bndMaxGen}
\max_{z\in[0,1]} \vert X(z) \vert^2 \le \vert X(0) \vert^2 + \|X\|_{\cH^1((0,1);\R^n)}^2.
\end{equation}
\end{prop}
\begin{proof}
For each $z \in [0,1]$, it is observed that
\begin{align*}
 &\vert X(z)\vert^2-\vert X(0)\vert^2 
= \int_0^z \frac{d}{ds}(\vert X (s)\vert^2)ds\\
& \quad = 2\int_0^z  (X(s))^\top \partial X (s)ds\\
&\quad \le \int_0^z  \vert X(s)\vert^2 ds + \int _0^z \vert \partial X(s)\vert^2ds \\
& \quad \le \|X\|_{\cL^2((0,1);\R^n)}^2 + \|\partial X\|_{\cL^2((0,1);\R^n)}^2\\
& \quad = \|X\|_{\cH^1((0,1);\R^n)}^2
\end{align*}
which gives the desired bound.
\end{proof}

\begin{rem}
Proposition~1 basically allows to get the bounds on $\cC^0$ norm of the state $X$ in terms of its $\cH^1$ norm. Then, one can work with Lyapunov functions which basically quantify $\cH^1$ norm of the state, and work with its derivative. In the literature, we see that Agmon's inequality \cite[Lemma~2.4]{KrstSmys08} is also used to get a bound on $\cL^\infty$ norm in terms of $\cH^1$ norm as remarked in \cite[Remark~4]{mironchenko2015construction}. The inequality \eqref{eq:bndMaxGen} is however different from the conventional Agmon's inequality \cite[Lemma~2.4]{KrstSmys08}.
\end{rem}

In the light of Proposition~\ref{prop:maxH1}, one can obtain the estimate \eqref{eq:maxEstGen} from the inequality \eqref{eq:bndMaxGen}, by ensuring that the control input $u$ is chosen such that for each $t \ge 0$:
\begin{itemize}
\item The solution $X(\cdot, t)$ belongs to $\cH^1((0,1);\R^n)$;
\item It holds that $\vert X(0,t)\vert$ and $\| X(\cdot,t) \|_{\cH^1((0,1);\R^n)}$ are bounded by the size of the disturbance $\|d_{[0,t]}\|_\infty$ plus some exponentially decreasing term in time.
\end{itemize}

To achieve these objectives, the use of static controllers of the form $u(t) = K y(t)$, will result in trajectories $X$ which are not differential with respect to spatial variable due to (possibly discontinuous) disturbances, and hence the solutions are not contained in $\cH^1((0,1);\R^n)$. To remedy this problem, we propose the use of dynamic controllers for stabilization.

\subsection{Using Dynamic Controller for $\cH^1$-regular Solutions}
\begin{figure}[!b]
\centering
\begin{tikzpicture}[xscale = 0.5, yscale = 0.6, circuit ee IEC,
every info/.style={font=\footnotesize}]
\draw (-0.5,0) node [rectangle, rounded corners, draw, minimum height =0.65cm, text centered] (sys) {$ \left\{ \begin{aligned}& \partial_t X (z,t) +\Lambda \partial_z X (z,t) = 0\\ &X(0,t) = HX(1,t) + Bu(t) \end{aligned}\right .$};
\draw (0,-5) node [rectangle, rounded corners, draw, minimum height =1cm, text centered] (obs) {\small $\left\{ \begin{aligned} \dot \eta (t)& = R\eta(t)  + S y(t) \\ u(t)&= K \eta (t)\end{aligned}\right .$};
\coordinate (upSamp) at ([xshift=1.5cm]sys.east);
\coordinate (bl) at ([xshift=-4cm]obs.west);
\coordinate (tr) at ([xshift=1.5cm]upSamp);
\coordinate (ftr) at ([xshift=2.5cm]tr);
\coordinate (br) at ([yshift=-5cm]tr);
\coordinate (tl) at ([yshift=5cm]bl);
\coordinate (upr) at ([yshift=-1.25cm]tr);
\coordinate (downr) at ([yshift=-3.5cm]tr);
\coordinate (upl) at ([yshift=-1.25cm]tl);
\coordinate (downl) at ([yshift=-3.5cm]tl);
\draw (tr) node [circle, draw, minimum height =0.65cm] (sum) {};

\draw [thick, ->] (sys.east) --node[anchor=south] {$X(1,t)$} (sum.west);
\draw [thick,->] (ftr) node[anchor=south east] {$d(t)$}-- (sum.east);
\draw [thick,-] (sum.south)--(br);
\draw [thick,->] (br) -- node[anchor=north] {$\qquad\qquad y(t) = X(1,t) + d(t)$} (obs.east);
\draw [thick, -] (obs.west) -- (bl);
\draw [thick,-] (bl)--(tl);
\draw [thick,->] (tl) -- (sys.west);
\end{tikzpicture}
\caption{Control architecture used for stabilization of hyperbolic system in the presence of disturbances.}
\label{fig:loopQuantHyp}
\end{figure}
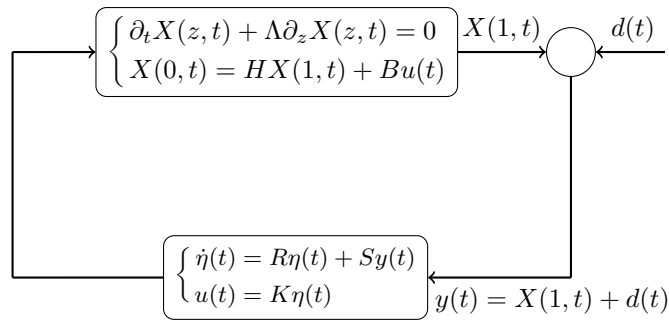
For system class \eqref{eq:sysHyp}, \eqref{eq:initCond}, \eqref{eq:defOut}, we are interested in designing control inputs $u$ that are absolutely continuous functions of time, so that their derivative is defined Lebesgue a.e.
For such inputs, we seek a solution $X \in \cC^0([0,T]; \cH^1((0,1);\R^n))$ where $\cC^0$ denotes the space of continuous functions equipped with supremum norm.

More precisely, we consider the problem of designing a dynamic controller with ODEs, which has the form
\begin{subequations}\label{eq:contGen}
\begin{align}
\dot \eta(t) &= R \eta (t) + S y(t) \label{eq:contGena}\\
u(t) & = K \eta(t) \label{eq:contGenb}
\end{align}
\end{subequations}
where the matrices $R \in \R^{n \times n}$, $S \in \R^{n \times n}$, and $K\in \R^{m \times n}$, need to be chosen appropriately.
Thus, the resulting closed loop is depicted in Figure~\ref{fig:loopQuantHyp}.
Intuitively speaking, by using such a controller, the discontinuities of the output $y$ are integrated via equation \eqref{eq:contGena} which results in $u$ being absolutely continuous.
The result on existence and uniqueness of solutions for the closed-loop system \eqref{eq:sysHyp}, \eqref{eq:initCond}, \eqref{eq:defOut}, \eqref{eq:contGen}, is formally developed in Section~\ref{sec:cont}.
Afterwards, in Section~\ref{sec:lyap}, we design the parameters of the controller \eqref{eq:contGen}, and derive conditions on the system and controller data which establish the DSS estimate \eqref{eq:maxEstGen}.

\section{Existence of Solutions}\label{sec:cont}

The objective of this section is to develop a result on existence and uniqueness of solutions for the closed-loop system \eqref{eq:sysHyp}, \eqref{eq:initCond}, \eqref{eq:defOut}, \eqref{eq:contGen}, demonstrated in Figure~\ref{fig:loopQuantHyp}.
Before presenting our result, we remark that the solutions of hyperbolic PDEs is a well-studied topic.
For the intermediate results, we refer the reader to \cite[Appendix~A]{BastCoro16} and \cite[Chapter~3]{CurtZwar95}.
In \cite{BastCoro16}, the authors first present results with $\cH^1$-regularity for the autonomous with $u = 0$.
The results for ODE coupled with hyperbolic PDE with $d=0$ with $\cL^2$ regularity are also proven.
However, in these works, with $d \in \cL^\infty$, which introduces certain discontinuities, the solutions with $\cH^1$-regularity are not discussed.
On the other hand, the well-posedness results are presented for systems with dynamics described by infinitesimal generators of continuous semigroups.

In this section, our contribution lies in presenting a result on well-posedness of the ODE-PDE coupled system of Figure~\ref{fig:loopQuantHyp}.
We do so by building on the results described in \cite[Appendix~A]{BastCoro16} and \cite[Chapter~3]{CurtZwar95}.


To do so, we start by constructing the operator $\cA$ as follows:
\begin{subequations}
\begin{gather}
\begin{split}
\dom (\cA) &:= \Bigg \{ (\varphi,\eta) \in \cH^1((0,1);\R^n) \times \R^n;\\
& \qquad\  \begin{pmatrix} \varphi (0) \\ \eta \end{pmatrix} = \begin{bmatrix} H & BK \\ 0 & I\end{bmatrix} \begin{pmatrix}\varphi(1) \\ \eta \end{pmatrix} \Bigg\},
\end{split} \\
\cA \begin{pmatrix} \varphi \\ \eta \end{pmatrix} := \begin{pmatrix}- \Lambda \varphi_z \\ R \eta \end{pmatrix}.
\end{gather}
\end{subequations}

Next, we introduce the {\em perturbation operator} $\cB : \cH^1((0,1);\R^n) \times \R^n \to \cH^1((0,1);\R^n) \times \R^n$ as follows:
\begin{subequations}
\begin{gather}
\cB \begin{pmatrix} \varphi \\ \eta\end{pmatrix} =  \begin{pmatrix} 0 \\ S \varphi (1) \end{pmatrix}.
\end{gather}
\end{subequations}

Using these operators $\cA$ and $\cB$, and letting $x = \col(X, \eta)$, one can write the closed-loop system \eqref{eq:sysHyp}, \eqref{eq:initCond}, \eqref{eq:defOut} and \eqref{eq:contGen} as follows:
\begin{subequations}\label{eq:deInfDim}
\begin{gather}
\dot x = \cA x + \cB x + \widetilde d\\
x(0) = x^0 \in \dom(\cA),
\end{gather}
\end{subequations}
where $\widetilde d = \begin{pmatrix} 0 \\ S \, d \end{pmatrix}$.
We now prove a result on the well-posedness of system \eqref{eq:deInfDim}. Because $d$ is possibly discontinuous, the classical solutions (where $\dot x$ is continuous) do not exist, and one must work with the notion of a weak solution \cite[Definition~3.1.6]{CurtZwar95}.

\begin{defn}[Weak Solution]\label{def:soln}
Let $d \in \cL^{\infty}([0,T];\R^n)$. If $(\cA+\cB)$ is an infinitesimal generator of a $\cC^0$-semigroup $\cT$, then we call $x:[0,T] \to \cH^{1}((0,1);\R^n) \times \R^n$ a {\em weak solution} to \eqref{eq:deInfDim} when the following two conditions hold:
\begin{itemize}
\item $x \in \cC^0([0,T];\cH^1((0,1);\R^n)\times \R^n)$, and
\item For each $g \in \cC^0([0,T];\cH^1((0,1);\R^n)\times \R^n)$
\[
\int_0^T \innProd{x(s)}{g(s)}\,ds + \int_0^T \innProd{\widetilde d(s)}{h(s)}\,ds 
+ \int_0^T \innProd{x^0}{h(0)}\,ds = 0
\]
where $h(t):= -\int_t^T \cT^*(s-t) g(s) \, ds$, $\cT^*$ is the adjoint of the operator $\cT$, and the inner product is with respect to $\cH^1((0,1);\R^n) \times \R^n$.
\end{itemize}
\end{defn}

The well-posedness of \eqref{eq:deInfDim} is now obtained from the results given in \cite[Chapter~3]{CurtZwar95}. To invoke these results, the operators $\cA$ and $\cB$ must satisfy certain conditions. The desired properties of these operators are listed in the lemmas that follow, and their proofs are provided in the Appendix.

\begin{lem}\label{lem:AInfGen}
The operator $\cA$ is an infinitesimal generator of a $\cC^0$-semigroup.
\end{lem}

The recipe used for proving this result is inspired by the development given in \cite[Appendix~A, Proof of Thm.~A.I]{BastCoro16}, that is, we show that the operators $\cA$ and its adjoint $\cA^*$ are quasi-dissipative, and that $\cA$ is closed and densely defined. The difference in the calculations arises due to the presence of $\eta$-dynamics and the domain of $\cA$ depends on the $\eta$-dynamics. This changes the construction of the adjoint operator as well. One then invokes a generalization of Lumer-Phillips theorem reported in \cite[Corollary~2.2.3]{CurtZwar95} to show that the operators with such properties are infinitesimal generators of strongly continuous semigroups.

\begin{lem}\label{lem:boundedB}
The operator $\cB$ is a bounded linear operator.
\end{lem}

The linearity of the operator $\cB$ is obvious from its definition. For showing boundedness, one should be careful in using the right norms, because $\cB$ embeds a finite dimensional vector in the space $\cH^{1}((0,1);\R^n) \times \R^n$.

The properties of the operators given in Lemmas~\ref{lem:AInfGen} and \ref{lem:boundedB} lead to the following result:

\begin{lem}\label{lem:A+B}
The operator $\cA + \cB$ is an infinitesimal generator of a $\cC^0$-semigroup.
\end{lem}

The proof then follows by invoking \cite[Theorem~3.2.1]{CurtZwar95}. In fact, the resulting semigroup has a lower triangular structure due to special form of $\cA$ and $\cB$.

Coming back to the system description~\eqref{eq:deInfDim}, we now invoke the properties of the operator $\cA + \cB$ listed in Lemma~\ref{lem:A+B}, and arrive at the following result using \cite[Theorem~3.1.7]{CurtZwar95}.

\begin{thm}\label{thm:sol}
For a given $T > 0$, and $d \in \cL^\infty([0,T];\R^n)$, there is a unique {\em weak} solution to system \eqref{eq:deInfDim}. Equivalently, for each $(X^0, \eta^0) \in \dom(\cA)$, the closed-loop system \eqref{eq:sysHyp}, \eqref{eq:initCond}, \eqref{eq:defOut} and \eqref{eq:contGen} has a unique {\em weak} solution in the space $\cC^0([0,T];\cH^1((0,1);\R^n) \times \R^n)$.
\end{thm}

\begin{rem}
The so-called compatibility conditions on the initial condition $(X^0,\eta^0)$, that are required for $\cH^1$-regularity are imposed by requiring that $(X^0,\eta^0)$ belong to $\dom(\cA)$. Such a condition is essential and hence the choice of $\eta^0$ depends upon $X^0$.
It is noted that in \cite{CoronBastin08}, the authors propose {\bf two} compatibility conditions for the initial state because they seek solution $X \in \cH^2 ((0,1);\R^n)$. We only need solutions where $X$ is $\cH^1$-regular, so only one such condition appears in our analysis.
\end{rem}


\begin{rem}\label{rem:sgProp}
It follows from the Definition~\ref{def:soln}, that the mild solution to equation~\eqref{eq:deInfDim} is given by:
\[
x(t) = \cT(t) \, x^0 + \int_0^t \cT(t-\tau) \, \widetilde d(\tau) \, d\tau.
\]
Thus, for any $t > s \ge 0$, we have
\[
x(t) = \cT(t - s) \, x(s) + \int_{s}^{t} \cT(t-\tau) \, \widetilde d(\tau) \, d\tau.
\]
\end{rem}

\section{Closed Loop and Stability Analysis}\label{sec:lyap}

As a solution to the problem formulated in Section~\ref{sec:prelim}, we now provide more structure for the controller dynamics, and study the stability of the closed-loop system.
The conditions on the system parameters that guarantee stability are then provided by constructing a Lyapunov-function.

\subsection{Control Architecture and Closed Loop}

The controller that we choose for our purposes is described by the following equations:
\begin{subequations}\label{eq:contLin}
\begin{align}
\dot \eta (t) & = -\alpha (\eta(t)- y(t)) \notag \\
& =  -\alpha \, \eta(t) + \alpha X(1,t) + \alpha d(t) \label{eq:contLina}\\
\eta(0)&=\eta^0\label{eq:init:eta}\\
u(t) & = K \eta(t), \label{eq:contLinb}
\end{align}
\end{subequations}
where $\eta^0 \in \R^n$ is the initial condition for the controller dynamics.
This corresponds to choosing $R = -\alpha\, I_{n\times n}$, and $S = -R$ in \eqref{eq:contGena}.
The conditions on the constant $\alpha >0$, and the matrix $K \in \R^{m\times n}$ will be stated in the statement of Theorem~\ref{thm:mainISS}.

For the system in the closed loop, the dynamics of the state trajectory $X$ are given by
\begin{subequations}\label{eq:closed:loop:X}
\begin{gather}
X_t (z,t) +\Lambda X_z (z,t) = 0, \label{eq:closed:loop:X:dynamics}\\
X(z,0)=X^0(z), \; \forall z\in [0,1],\\
X(0,t) = H X(1,t) + BK \eta (t).\label{eq:closed:loop:X:bd}
\end{gather}
\end{subequations}

For what follows, we are also interested in analyzing the dynamics of $\partial_z X =: X_z$ which are derived as follows:
\begin{equation}
\frac{\partial X_z}{\partial t} (z,t) +\Lambda \frac{\partial X_z}{\partial z} (z,t) = 0.
\end{equation}
To obtain the boundary condition, from (\ref{eq:closed:loop:X:bd}), we have
\[
X_t(0,t) = HX_t(1,t) + B K \dot \eta(t).
\]
Substituting $X_t(z,t) = -\Lambda X_z(z,t)$ for each $z \in [0,1]$, we get
\begin{equation}\label{eq:bCondXz}
X_z(0,t) = \Lambda^{-1} H \Lambda X_z(1,t) - \Lambda^{-1} B K \, \dot \eta(t).
\end{equation}

\begin{rem}
Note that the equation \eqref{eq:bCondXz} would be well defined for $X(\cdot,t) \in \cC^1((0,1);\R^n)$ with the obvious interpretation that $X_z(0,t):=\lim_{\eps \searrow 0} X_z(\eps,t)$, and $X_z(1,t):=\lim_{\eps \searrow 0} X_z(1-\eps,t)$. The same interpretation holds for \eqref{eq:closed:loop:X:bd}. In the sequel, when carrying out calculations in stability analysis, it will be assumed that $X(\cdot,t) \in \cC^1((0,1);\R^n)$, and by the density argument, the same conclusion would hold for $X(\cdot,t) \in \cH^1((0,1);\R^n)$, as done in \cite{CoronBastin08}.
\end{rem}

\subsection{Stability Result}
The second main contribution of the paper is to present conditions on the controller dynamics \eqref{eq:contLin} which results in stability of system~\eqref{eq:sysHyp} and robustness with respect to the measurement disturbances $d$.
To state the result, we introduce some notation. Let $\cD_+^n$ denote the set of diagonal positive definite matrices. For scalars $\mu>0$ and $0 < \nu < 1$, let $\rho := e^{-\mu} - \nu^2$; let $F:= BK$, and $Q:=F^\top D^2 F$ for $D \in \cD_+^n$; and finally, let $G:= H^\top D^2 F$.
We denote by $\Omega$ the symmetric matrix
\[
\begin{bmatrix}
\rho \beta_1 D^2 & -\beta_1(G + Q) & 0\\
* & 2 \alpha \beta_3 - (\beta_1+\alpha^2\beta_2) Q  & \beta_3 I+ \alpha \beta_2 G \\
* & * & (\rho D^2 + Q + G+G^\top)\beta_2
\end{bmatrix}
\]
in which $\alpha, \beta_1, \beta_2, \beta_3$ are some positive constants, and $*$ denotes the transposed matrix block. In the following statement, we denote the induced-Euclidean norm of a matrix $M$ by $\|M\|_2$.

\begin{thm}\label{thm:mainISS}
Assume that there exist scalars $\mu,\nu > 0$, a matrix $D \in \cD_+^{n}$, the gain matrix $K$, and the positive constants $\alpha,\beta_1,\beta_2,\beta_3$ in the definition of $\Omega$ such that
\begin{subequations}\label{eq:gainCond}
\begin{gather}
\| D(H+BK) D^{-1} \|_2 \le \nu < 1, \label{eq:gainConda}\\
\Omega > \zeta I \label{eq:gainCondb}
\end{gather}
\end{subequations}
for some scalar $\zeta > 0$.
Then, the closed-loop system satisfies the DSS estimate \eqref{eq:maxEstGen} with
\begin{equation}\label{eq:defMX0}
M_{X^0} := \|X^0 \|^2_{\cH^{1}((0,1);\R^n)} + \vert \eta^0 - X(1,0)\vert^2.
\end{equation}
\end{thm}

\begin{rem}[DSS implies ISpS]
For $M_{X^0}$ given in \eqref{eq:defMX0}, we obtain
\begin{align*}
M_{X^0} & \le  \|X^0\|^2_{\cH^{1}((0,1);\R^n)} + \vert X^0(1) \vert^2 + \vert \eta^0 \vert^2 \\
& \le \max_{z\in[0,1]} \vert X^0(z) \vert^2 + \|X^0\|_{\cH^{1}((0,1);\R^n)}^2 + \vert \eta^0 \vert^2
\end{align*}
where we used the obvious relation that $\vert X^0(1) \vert \le \max_{z\in[0,1]} \vert X^0(z) \vert$. 
Substituting this bound on $M_{X^0}$ in \eqref{eq:maxEstGen}, our DSS estimate leads to
\begin{multline}
\max_{z\in[0,1]} \vert X(z) \vert \le c \,\me^{-at} \max_{z\in[0,1]} \vert X^0(z) \vert^2 + \gamma (\|d_{[0,t]} \|_\infty) \\+ c \,\me^{-at} \left( \|X^0\|^2_{\cH^1((0,1);\R^n)} + \vert \eta^0\vert^2 \right).
\end{multline}
This is a more conventional input-to-state practical stability notion in $\cC^0([0,1];\R^n)$ with respect to disturbance $d \in \cL^{\infty}([0,\infty);\R^n)$.
In particular, with $d \equiv 0$, we have practical stability of $X = 0$ in $\cC^0([0,1];\R^n)$ in the following sense: For every $\eps > 0$, there exists $\delta > 0$ such that the following implication holds
\[
\max_{z\in[0,1]} \vert X^0(z) \vert^2 \le \delta \ \Rightarrow \ \max_{z\in[0,1]} \vert X(z,t) \vert \le \eps + C
\]
where $C:=\left( \|X^0\|^2_{\cH^1((0,1);\R^n)} + \vert \eta^0\vert^2 \right)$. 
\end{rem}

\begin{rem}
It must be noted that, if the initial condition of the closed-loop system $X^0, \eta^0$ is chosen such that $\|X^0 \|_{\cH^{1}((0,1);\R^n)} = 0$ and $\eta^0 = 0$, then $M_{X^0} = 0$.
Indeed, since $X^0 \in \cH^1((0,1);\R^n)$ implies that $X^0$ is continuous, and $\|X^0\|^2_{\cL^2((0,1);\R^n)} = 0$ implies that $X^0 = 0$ almost everywhere on $[0,1]$, we must have $X(1,0) = 0$.
\end{rem}

\begin{rem}
In the statement of Theorem~\ref{thm:mainISS}, condition \eqref{eq:gainConda} requires $\inf_{D\in\cD_+^n} \| D (H+BK) D^{-1} \|_2 < 1$ which also appears in the more general context of nonlinear systems \cite{CoronBastin08} when analyzing stability with respect to $\cH^2$-norm. However, the condition \eqref{eq:gainCondb} is introduced in our work to compensate for the lack of proportional gain in the feedback law. It definitely restricts the class of systems that can be treated with our approach and relaxing this condition or obtaining different criteria is a topic of further investigation.
\end{rem}

\begin{rem}
At this moment, we do not have a precise characterization of the parameters of system~\eqref{eq:sysHyp} for which \eqref{eq:gainCond} admits a solution.
As a particular instance, assume that \eqref{eq:gainConda} holds with $K = 0$.
In that case, the matrix $\Omega$ simplifies greatly as $Q = G = 0$.
Using the Schur complement, one can immediately find the constants $\alpha, \beta_1, \beta_2,\beta_3$ that result in $\Omega$ being positive definite, and hence satisfying \eqref{eq:gainCondb}.
By applying the continuity argument for solutions of matrix inequalities with respect to parameter variations, the solution to \eqref{eq:gainCondb} will also hold for $K \neq 0$, but sufficiently small.
\end{rem}

The remainder of this section is devoted to the proof of Theorem~\ref{thm:mainISS} and is divided into several steps for the ease of reading.
In Section~\ref{sec:lyapBounds}, we construct a function $V:\cH^1((0,1);\R^n) \times \R^n \to \R_+$.
By computing the derivative of this function in Section~\ref{sec:lyapDissip}, an upper bound on $\dot V$ along the solutions of (\ref{eq:contLin})-(\ref{eq:closed:loop:X}) is obtained under condition \eqref{eq:gainCond} which yields
\begin{equation}
\dot V (X(t),\eta(t)) \le -\sigma V(X(t),\eta(t)) +\chi \vert d(t)\vert^2
\end{equation}
for some constant $\sigma, \chi > 0$.
We then combine this bound with Proposition~\ref{prop:maxH1} in Section~\ref{sec:maxEst} to obtain the DSS estimate \eqref{eq:maxEstGen}.

\subsection{Construction of the Lyapunov Function}\label{sec:lyapBounds}
The primary idea is to introduce a Lyapunov function and analyze its derivative with respect to time.
As a candidate, we choose $V: \cH^{1}((0,1);\R^n) \times \R^n \to \R_+$ given by
\begin{equation}\label{eq:defLyap}
V := V_1 + V_2 + V_3
\end{equation}
where $V_1 : \cH^{1}((0,1);\R^n) \rightarrow \R_+$ is defined as,
\[
V_1 (X) := \int_0^1 X(z)^\top P_1 X(z) e^{-\mu z} \,dz,
\]
where $P_1$ is a diagonal positive definite matrix that will be specified later.
Similarly, $V_2 : \cH^{1}((0,1);\R^n) \rightarrow \R_+$ is given by
\[
V_2 (X) := \int_0^1 \partial X(z)^\top P_2 \partial X(z) e^{-\mu z}\,dz,
\]
where $P_2$ is a diagonal positive definite matrix that will be specified later, and finally $V_3:\cH^{1}((0,1);\R^n) \times \R^n \rightarrow \R_+$ is given by
\[
V_3 (X,\eta)= (\eta - X(1))^\top P_3 (\eta - X(1)),
\]
where $P_3$ is a symmetric positive definite matrix that will be specified later.

It is evident that there exist constants $\underline c_P:= \min_{i = 1,2,3}\{\lambda_{\min}(P_i)\}$, $\overline c_P:=\max_{i = 1,2,3}\{\lambda_{\max}(P_i)\}$ such that, for all $X\in \cH^1((0,1);\R^n)$, and $\eta \in \R^n$,
\begin{equation}\label{eq:lyapPropBnd}
\underline c_P (\|X\|_{\cH^1((0,1);\R^n)}^2 + \vert \eta-X(1)\vert^2) \le V (X,\eta) \le 
\overline c_P (\|X\|_{\cH^1((0,1);\R^n)}^2 + \vert \eta-X(1)\vert^2).
\end{equation}

\subsection{Lyapunov Dissipation Inequality}\label{sec:lyapDissip}
We now derive the bound on $\dot V$ that was used in Section~\ref{sec:lyapBounds} to obtain the desired ISS estimate.
This is done by analyzing the time derivative of each of the three functions in the definition of the Lyapunov function.

{\em Analyzing $V_1$:} Using an integration by parts and recalling that $P_1$ is a diagonal positive definite matrix, 
the time derivative of $V_1$ yields
\begin{align*}
\dot V_1 & = \int_0^1 (\partial_t X^\top P_1 X + X^\top P_1 \partial_t X) e^{-\mu z}\, dz \\
& = -\int_0^1 (\partial_z X^\top \Lambda P_1 X + X^\top P_1 \Lambda \partial_z X) e^{-\mu z} \, dz \\
& = -[X^\top \Lambda P_1 X e^{-\mu z}]_0^1  - \mu \int_0^1 X(z,t)^\top  P_1 \Lambda X(z,t) e^{-\mu z} dz \\
& \le - e^{-\mu} \, X(1,t)^\top P_1 \Lambda X(1,t) + X(0,t)^\top P_1 \Lambda X(0,t) - \sigma_1 V_1,
\end{align*}
where $\sigma_1=\mu \lambda_{\min}(\Lambda) $ in which $\lambda_{\min}(\Lambda)$ denotes the minimal eigenvalue of $\Lambda$.

We now impose the boundary conditions by substituting the value of control $u$ given in \eqref{eq:closed:loop:X:bd} to get
\[
X(0,t) = (H+BK)X(1,t) + BK(\eta - X(1,t))
\]
which results in
\begin{align*}
\dot V_1 & \le - \sigma_1 V_1 - e^{-\mu} X(1,t)^\top \Lambda P_1 X(1,t) 
 + X(1,t)^\top (H+BK)^\top \Lambda P_1 (H+BK)  X(1,t) \\
& \quad + 2 X(1,t)^\top (H+BK )^\top \Lambda P_1  BK (\eta- X(1,t)) 
 + (\eta-X(1,t))^\top K^\top B^\top \Lambda P_1 B K (\eta - X(1,t)).
\end{align*}
Pick a diagonal positive definite matrix $D$ such that (\ref{eq:gainConda}) holds. We check that 
\begin{equation}\label{eq:boundNuD2}
(H+BK)^\top D^2 (H+BK) \leq \nu^2 D^2\ .
\end{equation} 
Let $P_1 = \beta_1 D^2\Lambda^{-1}$ for some $\beta_1 > 0$. With \eqref{eq:boundNuD2}, we thus get
\begin{multline}
\dot V_1  \le -\sigma_1 V_1 - \beta_1(e^{-\mu} - \nu^2) X(1,t)^\top D^2 X(1,t) \\
+ 2 \beta_1 X(1,t)^\top (H+BK)^\top  D^2 BK (\eta - X(1,t)) \\
+ \beta_1 (\eta - X(1,t))^\top K^\top B^\top D^2 BK (\eta - X(1,t)).
\end{multline}
We will see in the sequel that the dynamic controller is chosen so that the last term vanishes in the analysis of $V_3$.

{\em Analyzing $V_2$:} Repeating the same calculations as in the case of $\dot V_1$, we get
\begin{align*}
\dot V_2  &\le - e^{-\mu}X_z(1,t)^\top P_2 \Lambda X_z(1,t)  + X_z(0,t)^\top P_2 \Lambda X_z(0,t)  - \sigma_2 V_2\ ,
\end{align*}
where $\sigma_2=\sigma_1$.
Using \eqref{eq:closed:loop:X:dynamics} and \eqref{eq:bCondXz}, the boundary condition for $X_z$ is rewritten as
\begin{align*}
X_z(0,t) &= \Lambda^{-1}(H+BK) \Lambda X_z(1,t) 
 -\Lambda^{-1}BK \left[\dot \eta - X_t(1,t)\right].
\end{align*}
Let $\widetilde D := D\Lambda$, then using (\ref{eq:gainConda}) again
\[
\| \widetilde D \Lambda^{-1}(H+BK) \Lambda \widetilde D^{-1} \|_2 \le \nu,
\]
 and we choose $P_2 = \beta_2 \widetilde D^2 \Lambda^{-1} $ to obtain 
\begin{multline*}
\dot V_2  \le - \sigma_2 V_2 - \beta_2  (e^{-\mu} - \nu^2) X_z(1,t)^\top \widetilde D^2 X_z(1,t) \\
 - 2 \beta_2 X_z(1,t)^\top \Lambda (H+BK)^\top \Lambda^{-1}\widetilde D^2 \Lambda^{-1}BK (\dot \eta - X_t(1,t)) \\
 + \beta_2 (\dot \eta - X_t(1,t))^\top K^\top B^\top \Lambda^{-1} \widetilde D^2 \Lambda^{-1}  B K (\dot \eta - X_t(1,t)).
\end{multline*}
It is observed that $\Lambda^{-1} \widetilde D^2 \Lambda^{-1} = D^2$ because $\Lambda$ and $\widetilde D = D \Lambda$ are diagonal matrices, so that
\begin{multline*}
\dot V_2  \le - \sigma_2 V_2 - \beta_2  (e^{-\mu} - \nu^2) X_t (1,t)^\top D^2 X_t(1,t) \\
 + 2 \beta_2 X_t(1,t)^\top (H+BK)^\top D^2 BK (\dot \eta - X_t(1,t)) \\
 + \beta_2 (\dot \eta - X_t(1,t))^\top K^\top B^\top D^2 B K (\dot \eta - X_t(1,t)),
\end{multline*}
where we also used \eqref{eq:closed:loop:X:dynamics} to write $X_z(1,t) = -\Lambda^{-1} X_t(1,t)$.
Substitute $\eta$-dynamics from \eqref{eq:contLina} and let $F:=BK$, to get
\begin{multline*}
\beta_2^{-1}\dot V_2 \le -\beta_2^{-1} \sigma_2 V_2 - (e^{-\mu} - \nu^2) X_t(1,t)^\top D^2 X_t(1,t) 
- 2\alpha X_t(1,t)^\top (H+F)^\top D^2 F (\eta - X(1,t)) \\
+ 2\alpha X_t(1,t)^\top (H+F)^\top D^2 F d(t) 
- 2 X_t(1,t)^\top (H+F)^\top D^2 F  X_t(1,t) \\
+ \alpha^2 (\eta - X(1,t))^\top F^\top D^2 F (\eta - X(1,t)) 
 +2\alpha(\eta - X(1,t))^\top F^\top D^2 F X_t(1,t)\\
+ X_t(1,t)^\top F^\top D^2 F X_t(1,t)
+ \alpha^2d(t)^\top F^\top D^2 F d(t) \\
-2 \alpha^2 (\eta - X(1,t))^\top F^\top D^2 F d(t) 
-2 \alpha X_t(1,t)^\top F^\top D^2 F d(t).
\end{multline*}
The terms involving $d$ appear in $\dot V_2$ because the $\eta$-dynamics are driven by the output which includes disturbances. Cancelation of certain terms yields
\begin{multline*}
\beta_2^{-1}\dot V_2 \le -\beta_2^{-1} \sigma_2 V_2 - (e^{-\mu} - \nu^2) X_t(1,t)^\top D^2 X_t(1,t) 
- 2\alpha X_t(1,t)^\top H^\top D^2 F (\eta - X(1,t)) \\
+ 2\alpha X_t(1,t)^\top H^\top D^2 F d(t) 
- 2 X_t(1,t)^\top H^\top D^2 F  X_t(1,t) \\
+ \alpha^2 (\eta - X(1,t))^\top F^\top D^2 F (\eta - X(1,t)) 
- X_t(1,t)^\top F^\top D^2 F X_t(1,t)\\
+ \alpha^2d(t)^\top F^\top D^2 F d(t) 
-2 \alpha^2 (\eta - X(1,t))^\top F^\top D^2 F d(t).
\end{multline*}
One can use the Young's inequality for the last two terms to decouple the disturbance from $X_t$ and $(\eta-X(1,t))$, that is, for every $\overline \zeta > 0$, we have
\[
2\alpha^2\beta_2(\eta - X(1,t))^\top F^\top D^2 F d(t) \le \overline \zeta |\eta - X(1,t)|^2 
+ \alpha^4\beta_2^2\frac{\|F^\top D^2 F \|_2^2}{\overline\zeta} |d(t)|^2
\]
\[
2\alpha\beta_2X_t(1,t)^\top H^\top D^2 F d(t) \le \overline\zeta |X_t(1,t)|^2 
+ (\alpha\beta)^2\frac{\|H^\top D^2 F \|_2^2}{\overline\zeta} |d(t)|^2.
\]

{\em Analyzing $V_3$:} Choose $P_3 = \beta_3 I$, and substitute the dynamics of $\eta$ from \eqref{eq:contLina} in the expression of $\dot V_3$ to obtain
\begin{align*}
\dot V_3 & = 2\,\beta_3 (\eta(t)-X(1,t))^\top(\dot \eta(t) - X_t(1,t)) \\
& = -2\alpha\beta_3\,|(\eta(t)-X(1,t))|^2 
 - 2 \beta_3 (\eta(t) -X(1,t))^\top X_t(1,t) 
 + 2 \alpha\beta_3 (\eta(t) -X(1,t))^\top d(t).
\end{align*}
Once again, Young's inequality is used to obtain, $\forall \, \overline \zeta > 0$
\[
2 \alpha\beta_3 (\eta(t) -X(1,t))^\top d(t) \le \frac{\overline \zeta}{2} \vert \eta(t) - X(1,t) \vert ^2 
+ \frac{2 (\alpha \beta_3)^2}{\overline \zeta} \vert d(t)\vert^2,
\]
which further yields
\begin{multline*}
\dot V_3 \le - \left(2 \alpha \beta_3 + \frac{\overline \zeta}{2} \right) \left\vert (\eta(t)-X(1,t)) \right\vert^2 
- 2 \beta_3 (\eta(t) -X(1,t))^\top X_t(1,t) \\
+ \overline \zeta |(\eta(t)-X(1,t))|^2 + \frac{2 (\alpha \beta_3)^2}{\overline \zeta} \vert d(t)\vert^2.
\end{multline*}

{\em Combining $\dot V_1, \dot V_2, \dot V_3$:}
By introducing the vector $w$ as
\[
w(t):= (X(1,t)^\top, (\eta(t)- X(1,t))^\top , X_t^\top(1,t))^\top,
\]
one can massage the terms in the expressions for $\dot V_i$, $i=1,2,3$
to get
\[
\dot V \le -\sigma_1V_1 - \sigma_2 V_2 - \frac{\overline\zeta}{2} V_3 - w^\top \Omega w
+ \overline \zeta w^\top w + \chi |d(t)|^2
\]
where the constant $\chi$ is given by
\begin{equation}\label{eq:defChi}
\chi := \alpha^4\beta_2^2\frac{\|F^\top D^2 F \|_2^2}{\overline\zeta} + (\alpha\beta)^2\frac{\|H^\top D^2 F \|_2^2}{\overline\zeta} 
+ \alpha^2 \|F^\top D^2 F \| + \frac{2 (\alpha \beta_3)^2}{\overline \zeta}.
\end{equation}
By choosing $\overline\zeta = \zeta$, where $\zeta$ satisfies \eqref{eq:gainCondb}, we obtain
\begin{equation}\label{eq:negLyap}
\dot V (X(t),\eta(t)) \le -\sigma V(X(t),\eta(t)) +\chi \vert d(t)\vert^2
\end{equation}
with $\sigma := \min \left\{\sigma_1, \sigma_2,\frac{\zeta}{2}\right\}$.

\subsection{Obtaining the DSS Estimate}\label{sec:maxEst}
For the $D \in \cD_n^+$ satisfying \eqref{eq:gainConda}, we apply the result of Proposition~\ref{prop:maxH1} to the function $DX(\cdot,t)$ to obtain the following estimate, for each $t \ge 0$:
\begin{equation}\label{eq:bndDX0}
\max_{z\in[0,1]} \vert DX(z,t) \vert^2 \le \vert DX(0,t) \vert^2 + \| DX(\cdot,t) \|_{\cH^{1}((0,1);\R^n)}.
\end{equation}
The boundary condition \eqref{eq:closed:loop:X:bd}, with $D \in \cD_+^n$, can be written as
\[
DX(0,t)= D(H+BK)D^{-1} DX(1,t) + DBK(\eta-X(1,t))
\]
which using Young's inequality and letting $F= BK$ yields
\begin{align*}
\vert D X(0,t) \vert^2 & \le \nu^2 \vert D X(1,t) \vert^2 + \|DF\|_2^2 (\eta - X(1,t))^2 \\
& \le \nu^2 \max_{z\in[0,1]} \vert D X(z,t) \vert^2 + \|DF\|_2^2 (\eta - X(1,t))^2.
\end{align*}
Substituting the last equation in \eqref{eq:bndDX0}, we get
\[
\max_{z\in[0,1]} \vert DX(z,t) \vert^2 \le
\frac{1}{1-\nu^2} \Big(\| DX(\cdot,t) \|^2_{\cH^{1}((0,1);\R^n)} + \|DF\|_2^2 (\eta - X(1,t))^2 \Big).
\]
Let us introduce the constant $c_D$ as
\[
c_D := \frac{\max\{\|D\|_2^2, \|DF\|_2^2\}}{\lambda_{\min}(D)^2(1-\nu^2)}
\]
then, for each $t \ge 0$:
\begin{align}
\max_{z\in[0,1]} \vert X(z,t) \vert^2 & \le c_D \left( \| X(\cdot,t) \|^2_{\cH^{1}((0,1);\R^n)} + \vert \eta - X(1,t)\vert^2 \right) \notag\\
& \le \frac{c_D}{\underline c_P} V(X(\cdot,t), \eta(t)) \label{eq:maxTimeBnd}
\end{align}
where we recall that $\underline c_P:= \min_{i = 1,2,3}\{\lambda_{\min}(P_i)\}$.

Next, by integrating \eqref{eq:negLyap}, we get
\begin{equation}\label{eq:lyapTimeBnd}
V(X(t),\eta(t)) \le \me^{-\sigma t} V(X(0),\eta(0)) + \frac{\chi}{\sigma} (\|d_{[0,t]}\|_\infty^2).
\end{equation}
To obtain the desired DSS estimate, we substitute the bound \eqref{eq:lyapTimeBnd} in \eqref{eq:maxTimeBnd} to get
\[
\max_{z\in[0,1]} \vert X(z,t) \vert^2 \le \frac{c_D \overline c_P}{\underline c_P} \me^{-\sigma \, t } \Big(\|X^0\|_{\cH^{1}((0,1);\R^n)}^2  +\vert \eta^0 - X(1,0)\vert^2 \Big) 
+ \frac{c_D\chi}{\underline c_P\sigma} \|d_{[0,t]}\|_\infty^2
\]
which is the desired DSS estimate~\eqref{eq:maxEstGen} with $M_{X^0}$ given in \eqref{eq:defMX0}.
This concludes the proof of Theorem~\ref{thm:mainISS}.

\subsection{ISS Estimate for the Closed Loop}
The DSS estimate \eqref{eq:maxEstGen} differs from the classical ISS estimate in the sense that we obtain a bound on the norm of $X(\cdot,t)$ in terms of the initial condition that depends on the state of the dynamical controller $\eta^0$.
However, if we consider the combined state of the closed-loop system $(X,\eta)$, then we can obtain a more conventional ISS estimate with this augmented state.
To see this, we observe that
\[
\vert \eta(t) \vert^2 \le 2 \vert \eta(t) - X(1,t) \vert^2 + 2 \max_{z\in[0,1]} \vert X(z,t)\vert^2
\]
and hence, from \eqref{eq:maxTimeBnd}, we have
\[
\vert \eta(t) \vert^2 + \max_{z\in[0,1]} \vert X(z,t)\vert^2 \le \frac{2+c_D}{\underline c_P} V(X(\cdot,t),\eta(t)).
\]
Once again, using the bound \eqref{eq:lyapTimeBnd}, we get
\begin{multline}
\|(X(\cdot,t),\eta(t))\|_{\cC^0([0,1];\R^n)\times \R^n}^2 \le C_1 \|d_{[0,t]}\|_\infty^2 \\
+C_2 \me^{-\sigma t} \left(2 \|(X^0,\eta^0)\|_{\cC^0([0,1];\R^n)\times \R^n}^2 + \|X^0\|^2_{\cH^1((0,1);\R^n)}\right)
\end{multline}
where $C_1:=\frac{(2+c_D)\chi}{\underline c_P\sigma}$ and $C_2:=\frac{(2+c_D)\overline c_P}{\underline c_P}$.
This is indeed a conventional ISS estimate for the closed-loop system with the state $(X,\eta)$, and $d$ viewed as an external disturbance.

\subsection{Effect of Vanishing Disturbance}\label{sec:vanish}
We now want to study the asymptotic behavior of the state $(X,\eta)$ when the disturbance $d$ is bounded and $d(t) \to 0$ as $t \to \infty$. For finite-dimensional systems, ISS estimates and the semigroup property of the solution set ensure that the corresponding state trajectories converge to zero asymptotically as $d$ converges to zero. We observe the same qualitative behavior with our DSS estimates.

From Remark~\ref{rem:sgProp}, where $x = (X,\eta)$, we recall that the solutions to the closed-loop system possess the semigroup property. If the estimate \eqref{eq:maxEstGen} holds, then for every $t > s \ge 0$, we have
\begin{equation}\label{eq:sgEstInf}
\max_{z\in [0,1]} \vert X(z,t) \vert \le c\, \me^{-a(t-s)} M_{X^s} + \gamma \left( \|d_{[s,t]}\|_\infty \right)
\end{equation}
where $M_{X^s} :=  \|X(s) \|^2_{\cH^{1}((0,1);\R^n)} + \vert \eta(s) - X(1,s)\vert^2$. From \eqref{eq:lyapPropBnd} and \eqref{eq:lyapTimeBnd}, it holds that $M_{X^s}$ is bounded by $\overline{c_P}M_{X^0} + \frac{\overline{c_P}\chi}{\sigma} \| d_{[0,s]} \|_{\infty}^2$.
Thus, in \eqref{eq:sgEstInf}, if $d(t) \rightarrow 0$ as $t \rightarrow \infty$, then by taking $s = t/2$, we see that $\max_{z\in [0,1]} \vert X(z,t) \vert \to 0$ as $t \to \infty$.

\section{Quantized Control}\label{sec:quant}

We are next interested in applying our results to study the stabilization of \eqref{eq:sysHyp}, where the measurement $X(1,t) \in \R^n$ is quantized, and cannot be transmitted to the control precisely.
In particular, $\R^n$-valued measurement $X(1,\cdot)$ is quantized using a finite set of alphabets, and hence the disturbances fed to the controller result from quantization error.
By working with uniform quantizer, we provide upper bounds on the number of symbols which result in the DSS estimate \eqref{eq:maxEstGen} with respect to a bounded quantization error.

\subsection{Description of the Quantizer}

To define a quantizer, we first specify a set of finite alphabets $\cQ:= \{q_0,q_1,q_2, \dots, q_N\}$, with $N$ chosen as an odd positive integer.
A quantizer with sensitivity $\Delta_q > 0$, and range $M_q>0$, is then a function $q:\R^n \rightarrow \cQ$ having the property that
\begin{equation}
\vert q (x) - x \vert_\infty \le \Delta_q \quad \text{\bf if} \quad \vert x \vert_\infty \le M_q
\end{equation}
and the overflow condition holds:
\begin{equation}
|q(x)|_\infty \ge M_q - \Delta_q \quad \text{\bf if} \quad \vert x \vert_\infty > M_q,
\end{equation}
where for $x:=\col(x_1, \dots, x_n) \in \R^n$, we used the notation $\vert x \vert_\infty := \max_{1 \le i \le n} |x_i|$.
Such a function $q$ defines what is called a {\em finite-rate uniform quantizer}.
In other words, within the space $\R^n$, where the measurements of $X(1,\cdot)$ take values, we take a cube with each side having length $2M_q$, and partition it uniformly in $N$ regions. 
Each of these regions is identified with a symbol $q_i$ from the set $\cQ$, $i \in \{1,\dots,N\}$.
If $\vert X(1,t) \vert_\infty \le M_q$, the controller receives a valid symbol $q_1, \dots, q_N$, and knows the variable $X(1,t)$, modulo the error due to sensitivity of the quantizer $\Delta_q$.
When the measurements are out of the range of the quantizer, that is, $\vert x \vert_\infty > M_q$, then the quantizer just sends an out of bounds flag $q_0$ and no upper bound on the error between $X(1,t)$ and its quantized value can be obtained in that case.

The cardinality of the set $\cQ$, or the number of regions, are determined by the ratio between the range and the sensitivity of the quantizer $M_q\slash \Delta_q$. This ration defines the rate at which the information is communicated by the quantizer on average. The basic idea of the quantized control in finite-dimensional systems is to show that the state of the system converges to a certain ball around the origin if this rate is sufficiently large (to dominate the most unstable mode) \cite{NairFagn07}.
In the same spirit, we derive a lower bound on the ratio $M_q\slash \Delta_q$ which is required to achieve practical stability in the presence of quantization errors.

\begin{rem}
Because the parameters $\Delta_q$ and $M_q$ remain constant in the definition of $q$, we are limiting ourselves to the case of static quantizers in this paper, that is, which results in a bounded measurement error determined by the sensitivity of the quantizer, if it can be ensured that $\vert X(1,t)\vert_\infty \le M_q$.
This is in contrast to the dynamic quantizers proposed in \cite{Libe03Aut} where the parameters $\Delta_q$ and $M_q$ are also updated while keeping their ratio constant, so that asymptotic stability of the origin could be achieved.
\end{rem}

\subsection{Stability Result with Quantized Control}

With quantized measurements, the controller \eqref{eq:contLin} takes the form
\begin{subequations}\label{eq:quantCont}
\begin{align}
\dot \eta (t) & = -\alpha \, \eta(t) + \alpha \, q(X(1,t)) \label{eq:quantConta}\\
u(t) & = K \eta(t).
\end{align}
\end{subequations}
By writing $q(X(1,t)) = X(1,t) + q(X(1,t)) - X(1,t)$, and letting $d_q(t):=q(X(1,t)) - X(1,t)$, we are indeed in the same setup as earlier with $y(t) = q(X(1,t))$. Here, $d_q$ is such that
\[
\vert d_q \vert \le \sqrt{n} \, \vert d_q \vert_\infty \le \sqrt{n} \, \Delta_q, \quad \text{if }\vert X(1,t) \vert_\infty \le M_q.
\]

\begin{thm}\label{thm:quant}
Consider the closed-loop system \eqref{eq:sysHyp} and \eqref{eq:quantCont}, and assume that the conditions \eqref{eq:gainConda} and \eqref{eq:gainCondb} hold.
Also, suppose that the initial conditions $X^0$ and $\eta^0$ satisfy
\begin{equation}\label{eq:boundIC}
\, V (X^0, \eta^0) \le \frac{\underline c_P}{c_D}M_q^2
\end{equation}
where $V(X,\eta)$ is defined in \eqref{eq:defLyap}.
With the constants  $\sigma,\chi$ appearing in \eqref{eq:negLyap}, if the quantizer is designed such that
\begin{equation}\label{eq:condQuant}
\frac{M_q^2}{\Delta_q^2} > \frac{n c_D \chi }{\underline c_P\sigma}.
\end{equation}
Then the following items hold:
\begin{itemize}
\item The output $X(1,t)$ remains within the range of the quantizer for all $t \ge 0$, that is,
\begin{equation}\label{eq:maxQuantBndX1}
\vert X(1,t) \vert_\infty \le M_q, \qquad \forall \, t \ge 0.
\end{equation}
\item The state of the system remains ultimately bounded in $\cC^0$-norm, that is, there exists $T$ such that for all $t \ge T$
\begin{equation}\label{eq:evQuantBndX}
\max_{z\in[0,1]} \vert X(z,t) \vert^2 \le \gamma_\eps(\Delta_q)
\end{equation}
where $\gamma_\eps$ is a class $\cK_\infty$ function
\[
\gamma_\eps(s) := \frac{n\chi c_D}{\sigma\underline c_P} s^2 (1+\eps),
\]
and $\eps > 0$ can be arbitrarily small.
\end{itemize}
\end{thm}
 

\begin{myproof}{of Theorem~\ref{thm:quant}}
In the light of condition~\eqref{eq:condQuant}, fix $\varepsilon > 0$ such that
\begin{equation}\label{eq:choiceEps}
\frac{n \chi}{\sigma} \Delta_q^2 (1+\varepsilon) < \frac{\underline c_P M_q^2}{c_D}.
\end{equation}
To proceed with the proof, we introduce two regions in the space $\cH^1((0,1);\R^n) \times \R^n$:
\[
\cS_M := \left\{ (X,\eta) \, \vert \, V(X,\eta) \le \frac{\underline c_P}{c_D}M_q^2 \right\}
\]
\[
\cS_\Delta := \left\{ (X,\eta) \, \vert \, V(X,\eta) \le \frac{n \chi}{\sigma} \Delta_q^2 (1+\varepsilon) \right\}.
\]
Because of \eqref{eq:choiceEps}, $\cS_\Delta$ is strictly contained inside $\cS_M$. We claim that the following two statements hold:
\begin{itemize}[leftmargin=4em]
\item[\em Claim~1:] If, for some $t_0 \ge 0$, $(X(t_0), \eta(t_0)) \in \cS_M \setminus \cS_\Delta $, then there exists a time $T_{\eps} \ge t_0$, such that $(X(T_{\eps}), \eta(T_{\eps})) \in \cS_\Delta$.
\item[\em Claim~2:] The set $\cS_M$ and $\cS_\Delta$ are forward invariant.
\end{itemize}
It is seen that the result of Theorem~\ref{thm:quant} holds because of these two claims. Since $\cS_M$ is invariant, and the initial condition $(X^0,\eta^0) \in \cS_M$ due to \eqref{eq:boundIC}, it follows that $(X(t),\eta(t))\in \cS_M$. We now invoke the inequality~\eqref{eq:maxTimeBnd} to observe that
\begin{equation}\label{eq:bndTrace}
\vert X(1,t) \vert_\infty^2 \le \vert X(1,t) \vert^2 \le \frac{c_D}{\underline c_P} \, V(X(\cdot,t),\eta(t))
\end{equation}
for all $t \ge 0$, which \eqref{eq:maxQuantBndX1} since $(X(t),\eta(t)) \in \cS_M$. To see that \eqref{eq:evQuantBndX} holds, it follows from Claim~1 and Claim~2 that for $t \ge T_\eps$, $(X(t),\eta(t)) \in \cS_\Delta$, and hence we have the desired bound on $\max_{z\in[0,1]}\vert X(z,t) \vert^2$ by making use of \eqref{eq:maxTimeBnd}.

{\em Proof of Claim~1:}
For $(X,\eta) \in \cS_M \setminus \cS_\Delta$, we compute $\dot V(X,\eta)$ along the closed-loop trajectories.
Because the measurement disturbance in $X(1,t)$ results from quantization error $d_q$, the derivative of the Lyapunov function in \eqref{eq:negLyap} satisfies
\[
\dot V(X,\eta) \le -\sigma V(X,\eta) + \chi \, d_q^\top d_q.
\]
For the region $\cS_M \setminus \cS_\Delta$, and for the chosen $\varepsilon > 0$,
\[
\frac{n \chi}{\sigma} \Delta_q^2 (1+\varepsilon) \le V (X,\eta) \le \frac{\underline c_P}{c_D}M_q^2.
\]
Also, if $(X, \eta) \in \cS_M$, then using \eqref{eq:bndTrace}, $\vert X(1)\vert_\infty \le \vert X(1)\vert \le M_q$ implying that $\vert d_q (t)\vert^2 \le n\Delta_q^2$, and hence
\begin{equation}\label{eq:epsNegLyap}
\dot V (X,\eta) \le -\varepsilon \frac{n \chi}{\sigma} \Delta_q^2.
\end{equation}
Thus, $V$ decreases strictly in the region $\cS_M\setminus\cS_\Delta$. Hence, if for some $t_0$, $(X(t_0),\eta(t_0)) \in \cS_M \setminus \cS_\Delta$, there exists a finite time $T_\eps \ge t_0$, such that $(X(T_\eps), \eta(T_\eps)) \in \cS_\Delta$.

{\em Proof of Claim~2:} Since $\cS_\Delta \subset \cS_M$, and every trajectory starting in $\cS_M\setminus\cS_\Delta$ reaches $\cS_\Delta$ due to Claim~1, it suffices to prove the forward invariance of $\cS_\Delta$ to establish the claim. Assume, for the sake of contradiction, that $\cS_\Delta$ is not forward invariant. Let $t_1$ be the {\em first} time instant such that
\[
V(X(t_1),\eta(t_1)) > \frac{n \chi}{\sigma} \Delta_q^2 (1+\varepsilon)
\]
Therefore, $(X(t),\eta(t)) \in \cS_M\setminus \cS_\Delta$ for each $t$ in a sufficiently small neighborhood of $t_1$. Hence, the inequality \eqref{eq:epsNegLyap} holds for $(X(t),\eta(t))$, for each $t$ near $t_1$. Thus, the absolutely continuous function $V(X(\cdot),\eta(\cdot))$ is negative definite in a neighborhood of $t_1$. Thus, $V(X(t),\eta(t)) > V(X(t_1),\eta(t_1))$ for some $t < t_1$. This contradicts the minimality of $t_1$, and hence the claim holds.
\end{myproof}

\subsection{Example}\label{sec:exQuant}

\begin{figure*}[t!]
    \centering
    \begin{subfigure}[t]{0.32\textwidth}
        \centering
        \includegraphics[width=\linewidth]{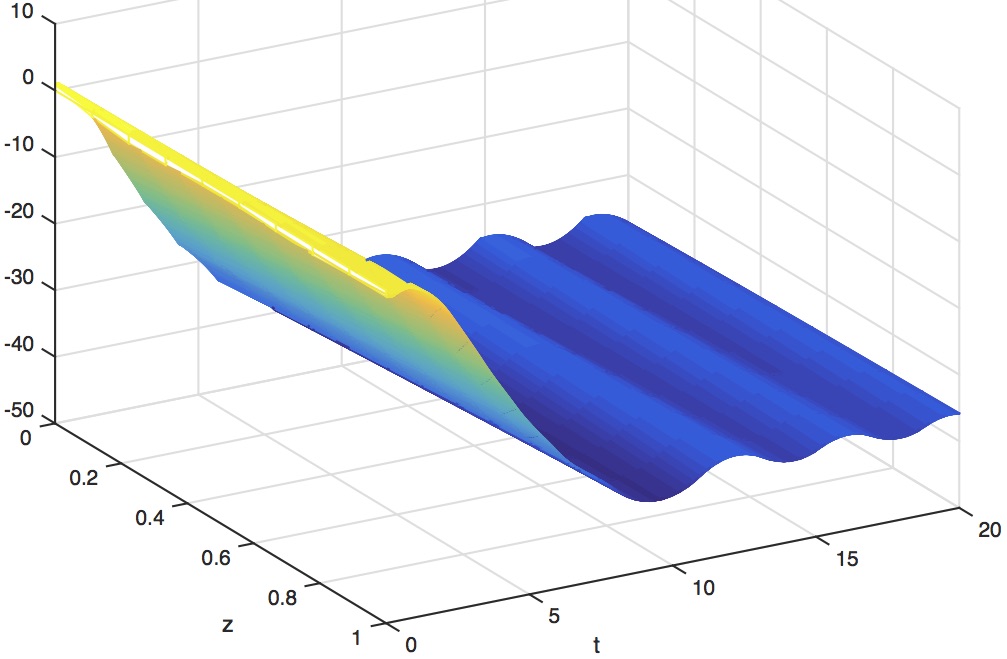}
        \caption{Evolution of $X_1$}
    \end{subfigure}%
    ~ 
    \begin{subfigure}[t]{0.32\textwidth}
        \centering
        \includegraphics[width=\linewidth]{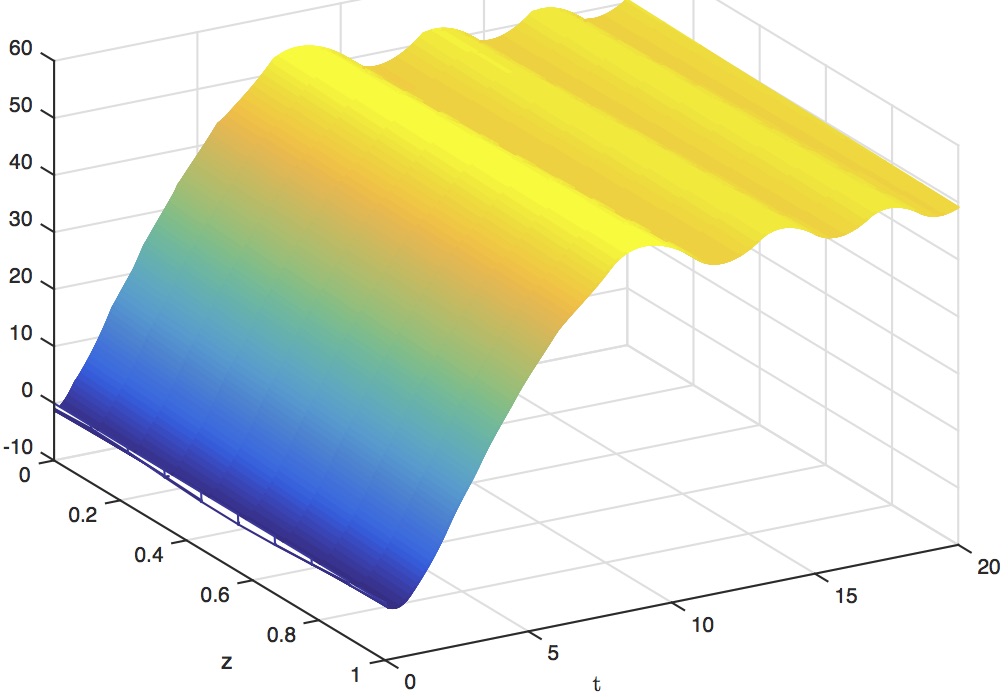}
        \caption{Evolution of $X_2$}
    \end{subfigure}
    ~
    \begin{subfigure}[t]{0.32\textwidth}
        \centering
        \includegraphics[width=\linewidth, height=4cm]{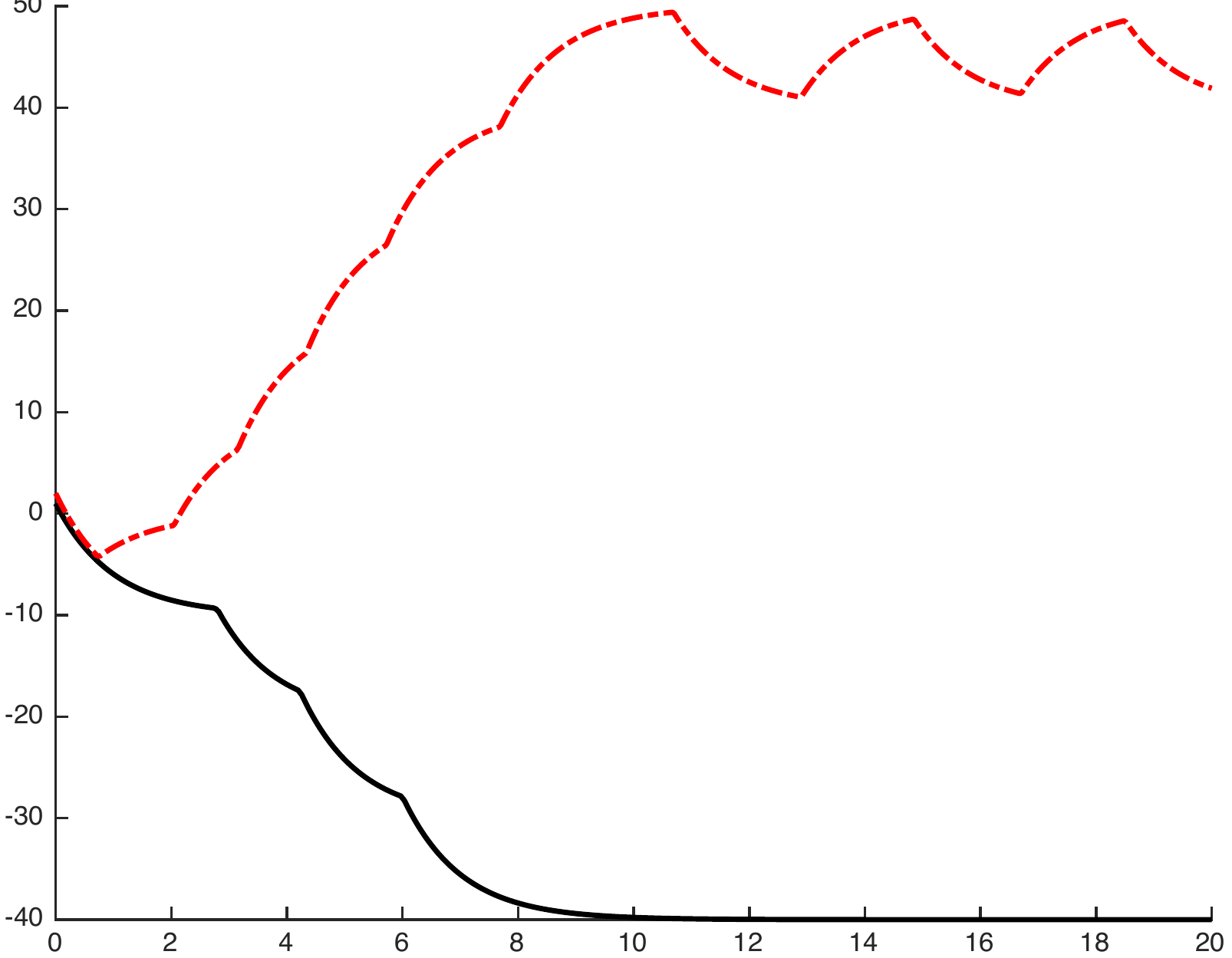}
        \caption{Black: $\eta_1$, Red: $\eta_2$}
    \end{subfigure}    
    \caption{Evolution of closed-loop trajectories with $\ell=0.1$.}
    \label{fig:ell005}
\end{figure*}

\begin{figure*}[t!]
    \centering
    \begin{subfigure}[t]{0.32\textwidth}
        \centering
        \includegraphics[width=\linewidth]{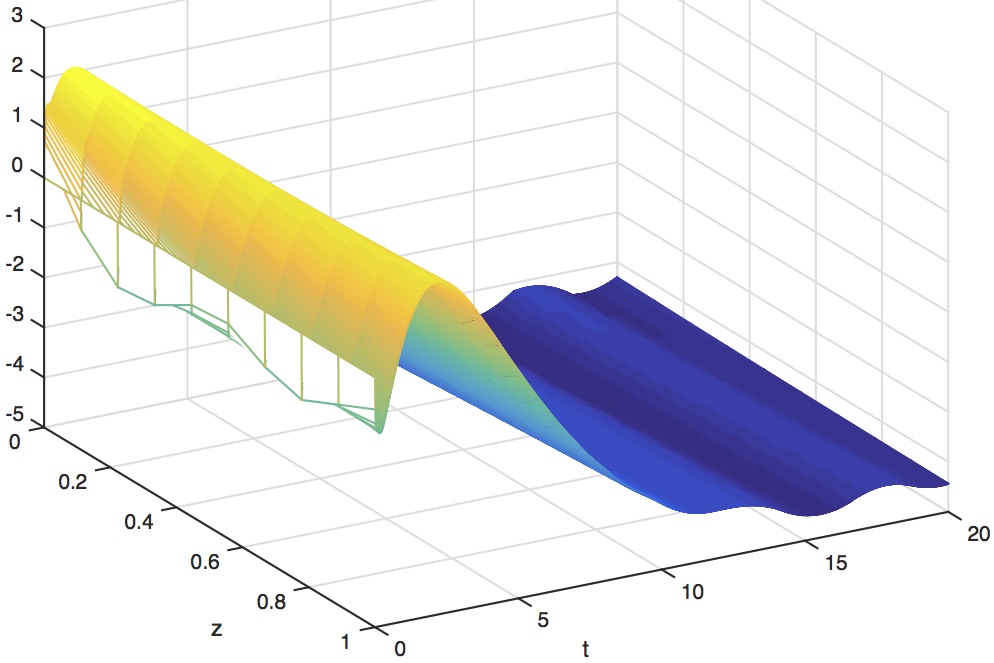}
        \caption{Evolution of $X_1$}
    \end{subfigure}%
    ~ 
    \begin{subfigure}[t]{0.32\textwidth}
        \centering
        \includegraphics[width=\linewidth]{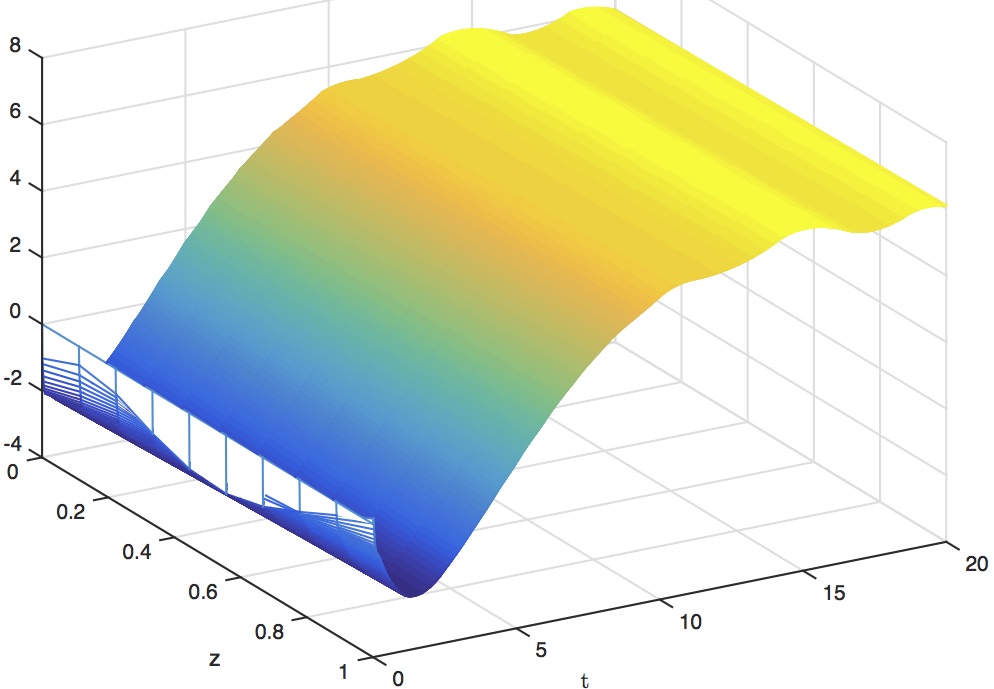}
        \caption{Evolution of $X_2$}
    \end{subfigure}
    ~
    \begin{subfigure}[t]{0.32\textwidth}
        \centering
        \includegraphics[width=\linewidth, height=4cm]{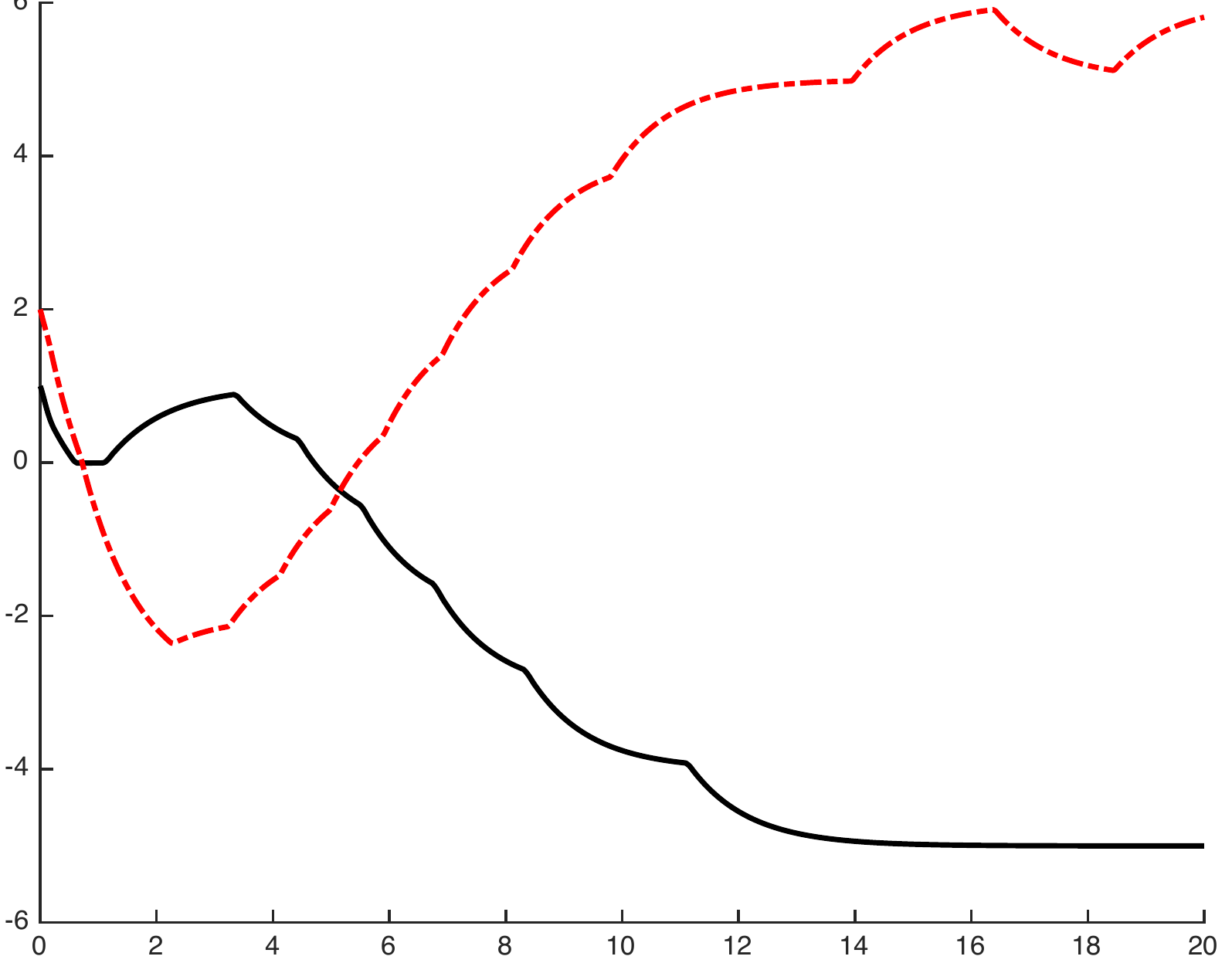}
        \caption{Black: $\eta_1$, Red: $\eta_2$}
    \end{subfigure}    
    \caption{Evolution of closed-loop trajectories with $\ell=1$.}
    \label{fig:ell05}
\end{figure*}

\begin{figure*}[t!]
    \centering
    \begin{subfigure}[t]{0.32\textwidth}
        \centering
        \includegraphics[width=\linewidth]{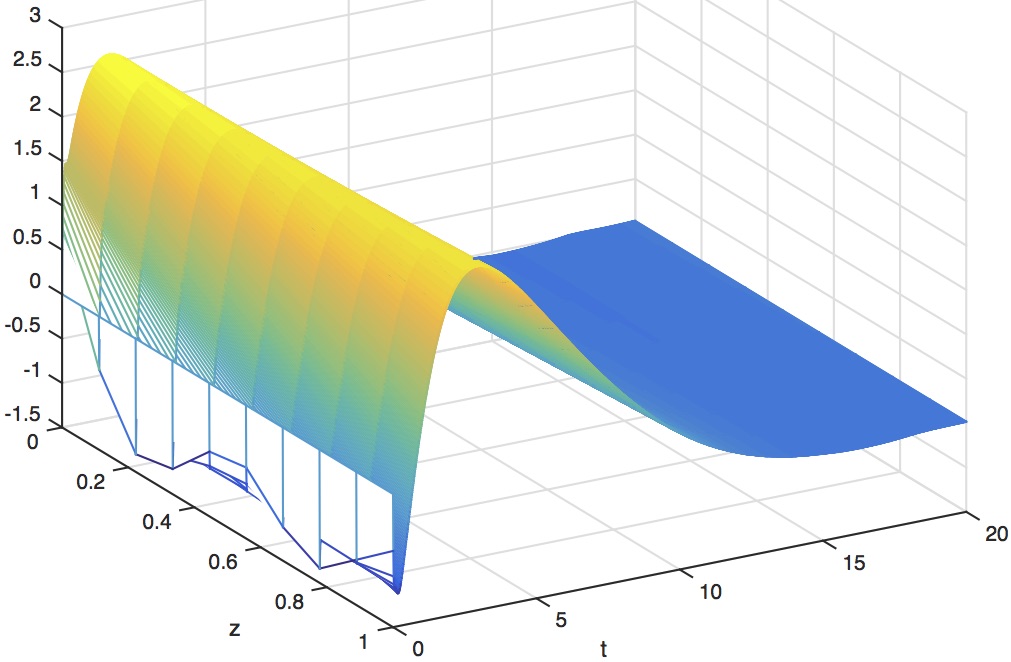}
        \caption{Evolution of $X_1$}
    \end{subfigure}%
    ~ 
    \begin{subfigure}[t]{0.32\textwidth}
        \centering
        \includegraphics[width=\linewidth]{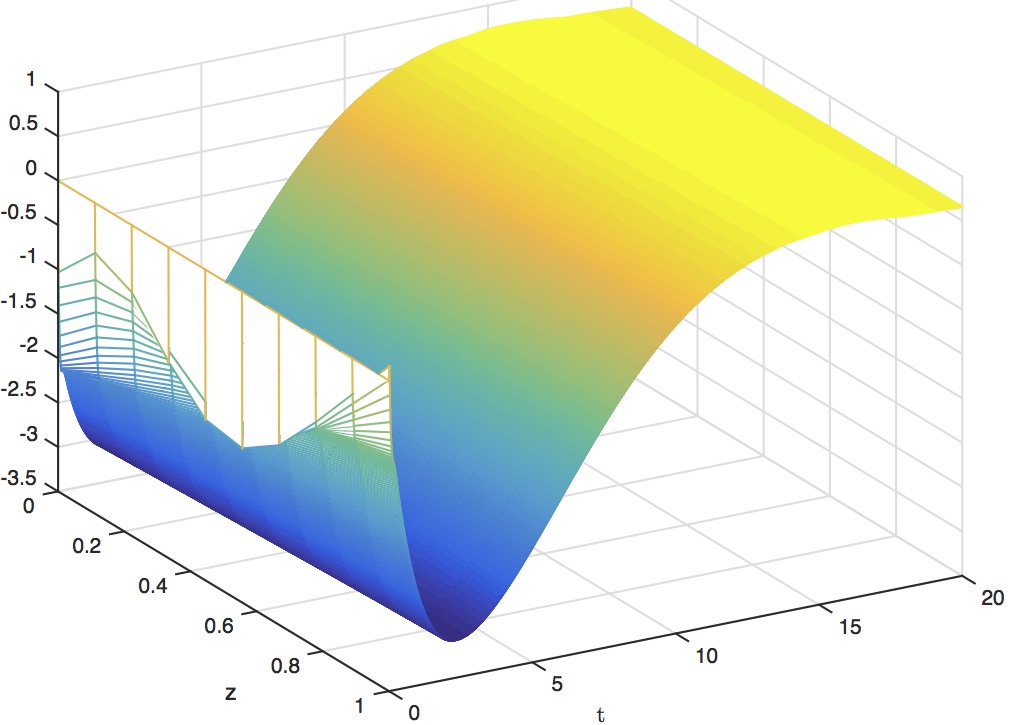}
        \caption{Evolution of $X_2$}
    \end{subfigure}
    ~
    \begin{subfigure}[t]{0.32\textwidth}
        \centering
        \includegraphics[width=\linewidth, height=4cm]{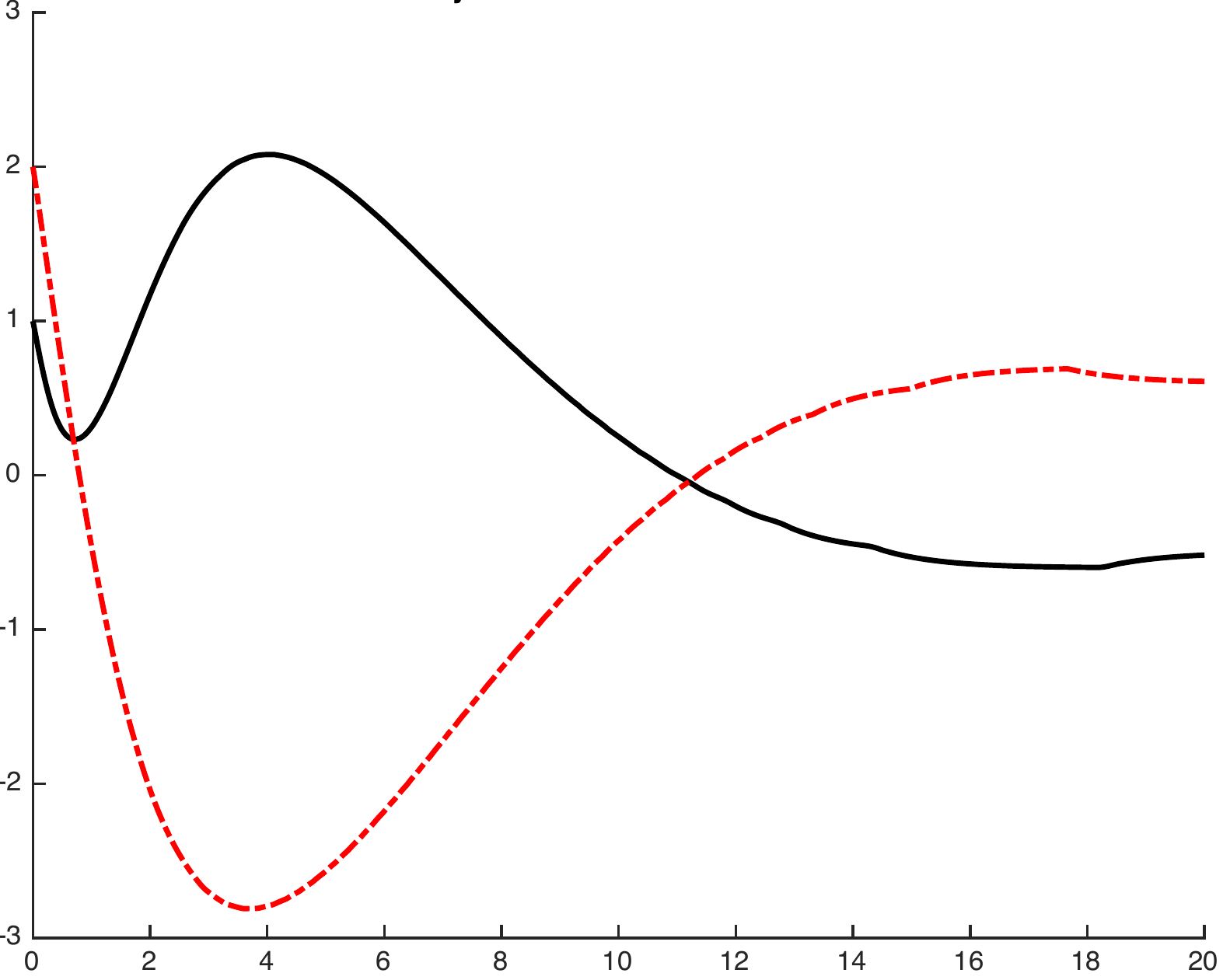}
        \caption{Black: $\eta_1$, Red: $\eta_2$}
    \end{subfigure}    
    \caption{Evolution of closed-loop trajectories with $\ell=10$.}
    \label{fig:ell5}
\end{figure*}

To illustrate the controller proposed in the previous section, the simulations for the case of a $2 \times 2$ hyperbolic system are now shown.
The system we simulate is of the form \eqref{eq:sysHyp} with
\[
\Lambda:=\begin{bmatrix} 1 & 0 \\ 0 & 2\end{bmatrix}
\]
and the boundary condition is described by 
\[
H = \begin{bmatrix} 0.25 & -1 \\ 0 & 1.25\end{bmatrix}, \quad B = \begin{bmatrix} 1 & 0 \\ 0 & 1\end{bmatrix}.
\]
Selecting the matrix $K = \begin{bmatrix} 0  & 0.5\\ -0.25& -0.5\end{bmatrix}$, it could be checked that the boundary damping condition \eqref{eq:gainCond} are satisfied, and thus the DSS estimate holds for \eqref{eq:sysHyp}-\eqref{eq:initCond} with the closed-loop boundary condition \eqref{eq:initCondLoop}.
Select the following initial condition, which satisfies the first-roder compatibility condition for the existence of solutions in $\cH^1((0,1);\R^n)$:
\[
X_1(z,0) =\cos (4\pi z)- 1 \; ,\quad X_2(z,0) =\cos (2\pi z)- 1,
\]
for $z \in [0,1]$.
 
Now to illustrate Theorem \ref{thm:quant}, let us consider the quantizer given by $q(x)= \floor{\ell x}/\ell $ with the parameter $\ell$.
The error due to quantization in this case is $\zeta_q = 1\slash \ell$, and for the sake of simplicity we take the range to be sufficiently large.

The time-evolution of the solutions for the first and second component of $X$, as well as the state of the dynamic controller $\eta$ are plotted in Figure~\ref{fig:ell005} for $\ell = 0.1$, and same entities are plotted in Figure~\ref{fig:ell05} and Figure~\ref{fig:ell5} for $\ell = 1$ and $\ell = 10$, respectively.
It could be seen that the solution to \eqref{eq:initCond} and \eqref{eq:quantCont} converges to a neighborhood of the origin as the time increases. The size of this neighborhood is seen to be decreasing as we increase the value of $\ell$, that is, the steady state values of $(X,\eta)$ are farther from the origin in Figure~\ref{fig:ell005} with $\ell=0.1$, compared to the steady state values of $(X,\eta)$ in Figure~\ref{fig:ell5} with $\ell = 10$. This is because, the upper bound on the error due to quantization of decreases as $\ell$ increases. These simulations are thus in agreement with the result reported in Theorem~\ref{thm:quant}.

%
       
\section{Conclusions}

We considered the problem of stabilization of boundary controlled linear hyperbolic PDEs in the presence of measurement errors in the output.
A notion of stability to describe robustness with respect to disturbances is introduced and a class of dynamic controllers is proposed under certain conditions which allow us to achieve this robust stability property.
We make connections of our proposed DSS notion with the conventional ISS and ISpS notions.
The results are used for an application when the output measurements are quantized over a finite alphabet set before being passed to the controller.
If the initial condition of the system is within the range of the quantizer, the resulting state trajectory is shown to converge to a ball parameterized by the quantization error.
Lower bounds on the cardinality of the alphabet set for the quantizer to achieve stability are also given.

Several interesting questions have come up in studying the problem. Firstly, we are interested in relaxing the stability condition that were presented in the statement of Theorem~\ref{thm:mainISS}.
One can also ask if adding nonlinear dynamics to the controller would lead to better results.

\section*{Acknowledgements}
The authors would like to thank Eduardo Cerpa for useful discussions related to the proof of well-posedness result, and anonymous reviewers for their numerous constructive remarks on an earlier version of this paper.

\appendix
\section{Proof of Lemma~\ref{lem:AInfGen}}

Consider the space $\cJ = \cL^2((0,1);\R^n) \times \R^n$. For some $\mu > 0$, this space is equipped with the inner product
\[
\innProd{\begin{pmatrix} \varphi \\ \eta \end{pmatrix}}{\begin{pmatrix} \psi \\ \theta \end{pmatrix}}\!_\mu.
=  \sum_{i=1}^n \int_0^1 \varphi_i \psi_i e^{\mu(z-1)} \, dz + \theta^\top \eta 
\]
The proof of this lemma builds on several intermediate steps.

{\em Step~1:} The operator $\cA$ is quasi-dissipative, that is, there exists a constant $C_\mu>0$ such that
\begin{equation}\label{eq:defQuasDiss}
\innProd{\cA \begin{pmatrix} \varphi \\ \eta \end{pmatrix}}{\begin{pmatrix} \varphi \\ \eta \end{pmatrix}}\!_\mu
\le C_\mu \innProd{\begin{pmatrix} \varphi \\ \eta \end{pmatrix}}{\begin{pmatrix} \varphi \\ \eta \end{pmatrix}}\!_\mu, \quad \forall \begin{pmatrix} \varphi \\ \eta \end{pmatrix} \in \dom(\cA).
\end{equation}
To see this, it is observed that
\begin{align*}
\innProd{\cA \begin{pmatrix} \varphi \\ \eta \end{pmatrix}}{\begin{pmatrix} \varphi \\ \eta \end{pmatrix}}\!_\mu
& = \sum_{i=1}^n \int_0^1 -\lambda_i (\varphi_i)_z \varphi_i e^{\mu(z-1)} \, dz + \eta^\top R \eta \\
& \le \frac{\mu}{2}\sum_{i=1}^n \int_0^1 \lambda_i \varphi_i^2 e^{\mu(z-1)} \, dz + \sum_{i=1}^n \lambda_i (\varphi_i^2(0)e^{-\mu} - \varphi_i^2(1))  + \|R\| \eta^\top \eta .
\end{align*}
Substituting the boundary condition $\varphi (0) = H \varphi(1) + BK \eta$, and using Young's inequality, we get
\[
\sum_{i=1}^n \varphi_i^2(0) \lambda_i e^{-\mu} \le c_1 \sum_{i=1}^n \varphi_i^2(1)\lambda_ie^{-\mu} + c_2 \sum_{i=1}^n \eta_i^2\lambda_ie^{-\mu}
\]
where $c_1 = \|H^\top H \|$ and $c_2 = \|K^\top B^\top BK\|$. Choose $\mu > 0 $ large enough such that
\[
\lambda_i c_1 e^{-\mu} \le 1, \quad \forall \, i \in \{1,\dots,n\}.
\]
Assuming that $\mu$ satisfies this condition, we thus obtain \eqref{eq:defQuasDiss} with
\[
C_\mu = \max\left\{1, \|R\|+ c_2 \lambda_{\max}(\Lambda) e^{-\mu}\right\},
\]
where $\lambda_{\max}(\Lambda)$ is the largest eigenvalue of the matrix $\Lambda$.

{\em Step~2:} The adjoint $\cA^*$ is quasi-dissipative.

By definition, $\cA^*$ is an operator that satisfies
\begin{equation}\label{eq:defAdj}
\left\langle \cA \begin{pmatrix} \varphi \\ \eta \end{pmatrix}, \begin{pmatrix} \psi \\ \theta \end{pmatrix}\right \rangle\!_\mu = \left\langle \begin{pmatrix} \varphi \\ \eta \end{pmatrix}, \cA^* \begin{pmatrix} \psi \\ \theta \end{pmatrix}\right \rangle\!_\mu \ .
\end{equation}
To compute $\cA^*$, we introduce the matrix $D \in \R^{n\times n}$
\[
D(z) := \diag \{e^{\mu(z-1)}, \dots, e^{\mu(z-1)}\}, \quad \forall \, z \in [0,1],
\]
and it is observed that
\begin{align*}
\left\langle \cA \begin{pmatrix} \varphi \\ \eta \end{pmatrix}, \begin{pmatrix} \psi \\ \theta \end{pmatrix}\right \rangle\!_\mu 
& = \int_0^1 -\varphi_z^\top \Lambda D \psi  \, dz + \theta^\top F\eta \\
& = - \left[ \varphi^\top \Lambda D \psi \right]_0^1 + \int_0^1 \varphi^\top \Lambda (D_z \psi + D\psi_z) \, dz + \eta^\top F^\top \theta \\
& \quad = \varphi(0)^\top\Lambda D(0) \psi(0) - \varphi(1)^\top\Lambda \psi(1)
 + \int_0^1 \varphi^\top \Lambda D (\mu \psi + \psi_z) dz + \eta^\top F^\top \theta\\
& \quad = \varphi(1)^\top H^\top\Lambda D(0) \psi(0) - \varphi(1)^\top\Lambda \psi(1) +\eta^\top F^\top \theta \\
& \qquad + \int_0^1 \varphi^\top D \Lambda (\mu \psi + \psi_z) dz +\eta^\top K^\top B^\top \Lambda D(0) \psi(0).
\end{align*}
Let $\cA^*$ be such that
\[
\dom(\cA^*) := \big\{\psi \in \cH^1((0,1);\R^n) \text{ such that } 
\psi(1) = \Lambda^{-1}H^\top \Lambda D(0) \psi(0)\big\}
\]
\[
\cA^*\begin{pmatrix} \psi \\ \theta \end{pmatrix} := 
\begin{pmatrix} \Lambda (\mu \psi + \psi_z) \\ K^\top B^\top \Lambda D(0) \psi(0) + F^\top \theta \end{pmatrix}.
\]
Clearly, with this definition of the adjoint operator, equation \eqref{eq:defAdj} holds. To show that $\cA^*$ is quasi-dissipative, we observe that
\begin{multline}\label{eq:formAdj}
\innProd{\cA^*\begin{pmatrix} \psi \\ \theta \end{pmatrix}}{\begin{pmatrix} \psi \\ \theta \end{pmatrix}}\! _\mu
 = \int_0^1 (\psi_z + \mu \psi)^\top \Lambda D \psi \, dz 
 + \theta^\top (K^\top B^\top \Lambda D(0) \psi(0) + F^\top \theta).
\end{multline}
Analyzing the first term on the right-hand side, we have
\begin{align}
& \quad \int_0^1 \psi_z^\top \Lambda D \psi \, dz \\
& = \frac{1}{2}\left[\psi^\top\Lambda D\psi \right]_0^1 + \frac{\mu}{2} \int_0^1 \psi^\top \Lambda D \psi \, dz \notag\\
& = \frac{1}{2} \left( \psi(1)^\top \Lambda \psi(1) - \psi(0)^\top \Lambda D(0)\psi(0) \right) + \frac{\mu}{2} \int_0^1 \psi^\top \Lambda D \psi \, dz \notag\\
& = \frac{1}{2} \psi(0)^\top \left( M_D M_H M_D - M_D \right)\psi(0) + \frac{\mu}{2} \int_0^1 \psi^\top \Lambda D \psi \, dz \label{eq:boundIntAdj}
\end{align}
where $M_D = \Lambda D(0)$ is a diagonal matrix, and $M_H = H\Lambda^{-1} H^\top$.

Next we observe that
\begin{equation}\label{eq:iniCondAdj}
\theta^\top K^\top B^\top \Lambda D(0) \psi(0) \le c_3 \theta^\top \theta + \frac{1}{2}\psi(0)^\top M_D^2\psi(0)
\end{equation}
where $c_3 = \frac{1}{2} \|BK\|^2$. Next, choose $\mu > 0$ such that,
\begin{equation}\label{eq:condMuAdj}
e^{-\mu} \lambda_{\max}(\Lambda) \,  \lambda_{\max}(M_H+I) \le 1.
\end{equation}
This condition ensures that $M_D(M_H+I) \le I$, and hence
\begin{align*}
M_DM_HM_D + M_D^2 - M_D
& = M_D (M_H + I) M_D - M_D \\
& \le M_D - M_D = 0.
\end{align*}
Substituting the expressions \eqref{eq:boundIntAdj}, \eqref{eq:iniCondAdj} in \eqref{eq:formAdj}, and using \eqref{eq:condMuAdj}, we get
\[
\innProd{\cA^*\begin{pmatrix} \psi \\ \theta \end{pmatrix}}{\begin{pmatrix} \psi \\ \theta \end{pmatrix}}\! _\mu
\le C_\mu \innProd{\begin{pmatrix} \psi \\ \theta \end{pmatrix}}{\begin{pmatrix} \psi \\ \theta \end{pmatrix}}\! _\mu
\]
where $C_\mu = \max\left\{ \frac{3}{2} \mu, c_3 + \|R\| \right\}$, and $\mu$ satisfies \eqref{eq:condMuAdj}.

{\em Step~3:} The operator $\cA$ is closed and $\dom(A)$ is dense in $\cL^2((0,1);\R^n)$.

To see that $\cA$ is closed, consider a sequence $(\varphi^k, \eta^k)$ such that
\[
\begin{pmatrix} \varphi^k \\ \eta^k\end{pmatrix} \to \begin{pmatrix}\varphi \\ \eta\end{pmatrix}
\quad \text{and} \quad 
\cA\begin{pmatrix} \varphi^k \\ \eta^k\end{pmatrix} = \begin{pmatrix} - \Lambda \varphi_z^k \\ R \eta^k \end{pmatrix} \to \begin{pmatrix}\psi \\ \theta\end{pmatrix}
\]
Since $\varphi, \psi \in \cL^2((0,1);\R^n)$, it follows that
\[
\varphi^k \to \varphi \quad \text{in } \cH^1((0,1);\R^n)
\]
and hence $\psi = -\Lambda \varphi_x$. The matrix $R$ defines a finite-dimensional linear operator, which yields $\theta = R\eta$. Also, $\cH^{1}((0,1);\R^n)$ is continuously embedded in $\cC^0([0,L];\R^n)$, and by picking $\varphi^k \in \dom(A)$, we get
\[
\begin{pmatrix} \varphi(0) \\ \eta \end{pmatrix} = \begin{bmatrix} H & BK \\ 0 & I \end{bmatrix} \begin{pmatrix} \varphi(1) \\ \eta \end{pmatrix} \ \in \dom(\cA).
\]

\section{Proof of Lemma~\ref{lem:boundedB}}

The linearity of the operator $\cB$ is obvious. To show that $\cB$ is bounded, we find a constant $C_\cB$ such that
\begin{equation}\label{eq:bndBProof}
\begin{aligned}
\left\| \cB \begin{pmatrix} \varphi \\ \eta \end{pmatrix} \right \|_{\cH^1((0,1);\R^n)\times \R^n} & = \| S \varphi(1) \|_{\R^n}  \le C_\cB \left\|\begin{pmatrix} \varphi \\ \eta \end{pmatrix} \right \|_{\cH^1((0,1);\R^n)\times \R^n} .
\end{aligned}
\end{equation}
To obtain such an inequality, we first observe that
\begin{align}
\vert \varphi(1) \vert^{2}  &= \left( \Big\vert \int_{0}^{1}s \varphi_z(s)  + \varphi(s) ds   \Big \vert \right)^{2} \notag \\
& \leq \left(  \int_{0}^{1} \vert s\varphi_z(s)\vert ds   + \int_{0}^{1} \vert \varphi(s) \vert ds    \vert \right)^{2} \notag \\
& \leq  2\left(\int_{0}^{1} \vert s\varphi_z(s)\vert ds  \right)^{2} +2 \left( \int_{0}^{1} \vert \varphi(s) \vert ds \right)^{2}\notag \\
& \leq  2 \left(\int_{0}^{1} \vert \varphi_z(s)\vert^{2} ds +   \int_{0}^{1} \vert \varphi(s) \vert^{2} ds  \right)\notag\\
&= 2\Vert \varphi \Vert^2_{\cH^{1}((0,1);\R^{n})}. \label{eq:bndPhi1}
\end{align}
%
from where the inequality in \eqref{eq:bndBProof} is obtained.

\bibliographystyle{plain}

\end{document}